\documentclass[a4paper,11pt]{article}

\usepackage{setspace}
\usepackage{amsfonts,amsthm,amssymb}
\usepackage{amssymb}
\usepackage[leqno]{amsmath}
\def \Rbrack {[\![}
\def \Lbrack {]\!]}

\newcommand\COMP{\hbox{C\kern -.58em {\raise .54ex \hbox{$\scriptscriptstyle |$}}
\kern-.55em {\raise .53ex \hbox{$\scriptscriptstyle |$}} }}
\newcommand\NN{\hbox{I\kern-.2em\hbox{N}}}
\newcommand\RR{{\it \hbox{I\kern-.2em\hbox{R}}}}
\newcommand\sRR{{\it \hbox{I\kern-.2em\hbox{R}}}}

\newcommand\ZZ{{{\rm Z}\kern-.28em{\rm Z}}}
\newcommand\be{\begin{equation}}
\newcommand\ee{\end{equation}}

\newtheorem{theorem}{Theorem}[section]
\newtheorem{assumption}[theorem]{Assumptions}
\newtheorem{proposition}[theorem]{Proposition}
\newtheorem{lemma}[theorem]{Lemma}

\newtheorem{corollary}[theorem]{Corollary}
\newtheorem{definition}[theorem]{Definition}
\newtheorem{remark}[theorem]{Remark}
\newtheorem{remarks}[theorem]{Remarks}
\newtheorem{examples}[theorem]{Examples}

\begin{document}

\title{Explicit Description of HARA Forward Utilities and Their Optimal Portfolios\thanks{This research is totally financed by the Natural Sciences
and Engineering Research Council of Canada (NSERC) through Grant G121210818. \newline The authors are very grateful to Freddy Delbaen and Jun Deng for their valuable advices and suggestions that helped improving the paper, and to Christoph Frei for informing them about Anthropelos (2013).}}

\author{Tahir Choulli\thanks{Corresponding author, Tel.: 1 780 492 9078, Fax: 1 780 492 6826, tchoulli@ualberta.ca} and Junfeng Ma\\
\\
Mathematical and Statistical Sciences Department \\
              University of Alberta, Edmonton\\
              Alberta, T6G 2G1, Canada}          
%

\maketitle

\begin{abstract}
This paper deals with forward performances of HARA type. Precisely, for a market model in which stock price processes are modeled by a locally bounded $d$-dimensional semimartingale, we elaborate a complete and explicit characterization for this type of forward utilities. Furthermore, the optimal portfolios for each of these forward utilities are explicitly described. Our approach is based on the minimal Hellinger martingale densities that are obtained from the important statistical concept of Hellinger process. These martingale densities were introduced recently, and appeared herein tailor-made for these forward utilities. After outlining our parametrization method for the HARA forward, we provide illustrations on discrete-time market models. Finally, we conclude our paper by pointing out a number of related open questions.
\end{abstract}


\section{Introduction}
Since the seminal papers of Merton (1971, 1973), the theory of utility maximization and optimal portfolio has been developed successfully in many directions and in different frameworks. These achievements can be found in the works of Karatzas and Wang (2000), Kramkov and Schachermayer (1999), Cvitanic et al. (1992, 2001), Karatzas and Zitkovic (2003), and the references therein to cite few. In these works, the authors considered a fixed investment horizon and practically neglected the impact of a variable horizon on the optimal selection portfolio and/or investor's behavior. The economic problem of how a horizon will impact an investment is old and can be traced back to Fisher (1931). In mathematical context, this problem is very difficult and only recently there were some advances. However, the problem where the agent should find an optimal portfolio from her investment in the stock market and the optimal time to liquidate all her assets (tradable or not tradable) was around since a while and has been addressed in many ways. For the literature about this problem, we refer the reader to Evans et al. (2008), Henderson and Hobson (2007), and the references therein. A particular interesting approach for this problem was proposed by Henderson and Hobson (2007) where the authors proposed to transfer the problem of investment and liquidation time to the problem of finding a utility that is not sensitive to the horizon. The authors called this utility functional, when it exists, a horizon-unbiased utility. This is one of the ideas that contributed to the birth of {\it forward utilities}. Indeed, the notion of forward utilities (or forward performances) appeared in the literature in various forms through mainly the works of Musiela--Zariphopoulou (2007, 2009a, 2009b, 2010), the work of Henderson--Hobson (2007), and the work of Choulli et al. (2007) (see Choulli and Stricker (2005, 2006) for other related topics). The forward utilities constitute a subclass of the large class of random field utilities that was brought to mathematical finance by Karatzas and Zitkovic (2003). These random field utilities appeared first in economics within the {\it random utility model theory} due to the psychometric literature that provided empirical evidence about the stochastic choice behavior. For details about these themes, we refer the reader to Suppes et al. (1989), McFadden and Richter (1970), Cohen (1980), and Clarck (1996) and the references therein to cite few. A random field utility represents the preference of an agent (or the agent's impatience as called in Fisher (1930)), which is updated at each instant using the available aggregate flow of public information about the market.  Among these random field utilities, forward utility has the feature of being a supermartingale for the wealth process generated by an admissible and self-financing strategy, while it is a martingale at the optimum. \\

\noindent After its birth, the notion of forward utilities has been extensively studied in different context and generalized in many ways. Indeed, forward utilities were used  in the context of risk measures in  Zariphopoulou and Zitkovic (2010). In Henderson (2007) and Anthropelos (2013), the notion of forward utilities is mainely and extensively used for indifference pricing and/or evaluation. There were many attempts to characterize these forward utilities starting with Berrier, Rogers and Tehrenchi (2009) under strong assumptions on the market models. Afterwards, Zitkovic (2009) elaborated a duality characterization for the semimartingale market model with explicit characterization when the market model is driven by Brownian motions and the utility is of exponential-type. At last, in Choulli et al. (2011), the forward utility of exponential type is completely and explicitly described in the semimartingale framework.\\

\noindent The main aim of this paper is to elaborate an explicit parametrization for forward utilities of HARA type. Precisely, we are interested in forward utilities, $U(t,x)$, having the form of
\begin{equation}\label{Theclass}
U(t,\omega,x)=D(t,\omega)x^{p(\omega,t)},\ \ \ \ \ \ \mbox{or}\ \ \ \  U(t,\omega,x)=\overline{D}(t,\omega)+{\widehat D}(t,\omega)\log(x).\end{equation}
We will describe explicitly the dynamics of the processes $D$, $p$, $\overline{D}$ and ${\widehat D}$ as well as the optimal portfolios for each class of forward utility. This will be achieved due to the concept of minimal Hellinger martingale densities introduced and developed in Choulli et al. (2007). The results of the actual paper ---Sections 3 and 4--- can not be put together with those of Choulli et al. (2011) in a unified framework. This fact is obvious from the different forms found for the risk-aversion processes in this current paper and in Choulli et al. (2011) respectively. The concept of Hellinger process/integral appeared in statistics and/or information theory, where it plays important r\^oles. It was extended and then slightly modified, for a better use in mathematical finance, by many scholars such as Jacod, Kabanov, Shiryaev, Stricker and Choulli. For more details, we refer the reader to Kabanov (1985), Kabanov et al. (1984, 1986), Chouli and Stricker (2005, 2006) and the references therein.\\

This paper is organized as follows. Section 2 will introduce the mathematical model as well as some preliminaries and notation. In Sections 3 and 4, we will detail our parametrization for the HARA forward utilities, while Section 5 illustrates the obtained characterizations on discrete-time market models.  The paper concludes the study by signaling some related open problems. Other technical and/or intermediatory results are gathered in the Appendix.

\section{Preliminaries and Notation}
This section contains two subsections. The first subsection presents the mathematical model as well as its preliminary analysis. The second subsection introduce the economical concepts of random field utility and forward utilities.
\subsection{The Mathematical Model}
The mathematical model starts with a given filtered probability
space denoted by $\left(\Omega,{\cal F}, {\mathbb F},P\right)$ where the filtration $\mathbb F:=({\cal F}_t)_{0\leq t\leq
T}$ is complete and right continuous,
and $T$ represents a fixed horizon for investments. In this setup,
we consider a $d$-dimensional {\bf locally bounded} semimartingale $S=(S_t)_{0\leq t\leq
T}$ which represents the discounted price processes of $d$ risky
assets.

Next, we recall the definition of the predictable characteristics of
the semimartingale $S$ (see Section II.2 of Jacod and Shiryaev (2003)).
 The random
measure $\mu$ associated to its jumps is defined by
 $$
 \mu(d t,\ d x)=\sum I_{\{\Delta S_s\not=0\}}\delta_{(s,\ \Delta S_s)}(d t,\ d x),
 $$
 with $\delta_a$ the Dirac measure at point $a$.
The continuous local martingale part of $S$ is denoted by $S^{c}$.
 This leads to the following decomposition, called ``{\it the
 canonical representation}'' (see Theorem 2.34, Section II.2 of Jacod and Shiryaev (2003)), namely,
\begin{equation}\label{modelS}
 S=S_0+S^{c}+x\star (\mu-\nu)+B,
 \end{equation}
 where the
random measure $\nu$ is the compensator of the random measure $\mu$. The entries of the matrix $C$ are $C^{ij}:=\langle S^{c,i},
S^{c,j}\rangle $, and the triple $(B,\ C,\ \nu)$ is called {\it predictable characteristics} of $S$.
  Furthermore, we can find a version of the characteristics triple satisfying
\begin{equation}\label{modelSbis} B=b\cdot A,\ \ C=c\cdot A\ \ \mbox{and}\ \
\nu(\omega,\ d t,\ d x)=d A_t(\omega)F_t(\omega,\ d x).
\end{equation} Here $A$ is an increasing and predictable process, $b$
and $c$ are predictable processes,
 $F_t(\omega,\ d x)$ is a predictable kernel, $b_t(\omega)$ is a vector in $\hbox{I\kern-.18em\hbox{R}}^d$ and
$c_t(\omega)$ is a symmetric $d\times d$-matrix , for all $(\omega,\
t)\in\Omega\times [0,\ T]$. In the sequel we will often drop
$\omega$ and $t$ and write, for instance, $F(d x)$ as a shorthand
for $F_t(\omega , d x)$.

\noindent The characteristics, $B,\ C,$ and $\nu$,  satisfy
$$
\begin{array}{l}
F_t(\omega,\ \{0\})=0,\hskip 0.6cm \displaystyle{\int} (\vert
x\vert^2\wedge 1)F_t(\omega,\
d x)\leq 1,\hskip 0.6cm \Delta B_t=\displaystyle{\int} x\nu(\{t\}, d x), \\
\\
 c=0\ \ \mbox{ on }\ \{\Delta A\neq 0\} .
 \end{array}
$$
We set
$$
\nu_t(d x):=\nu (\{t\}, d x),\ \  a_t:=\nu_t
(\hbox{I\kern-.18em\hbox{R}}^d)=\Delta A_t
F_t(\hbox{I\kern-.18em\hbox{R}}^d)\leq 1.
$$

\noindent We denote by $ \mathbb{P}_a$ (respectively $\mathbb{P}_e$) the set of all probability measures that are absolutely
continuous with respect to (respectively equivalent to) $P$.
 The set of local martingales
 under the probability $Q$ is denoted by ${\cal M}_{loc}(Q)$, while the set of equivalent martingale measures for $S$ is given by
  \begin{equation}\label{non-arbitrage}
  {\cal M}_{loc}^e(S):=\Bigl\{ Q\in \mathbb{P}_e:\ \ S\ \mbox{is a local martingale under}\ Q\Bigr\}.\end{equation}

\noindent As usual, ${\cal A}^+$ denotes the set of increasing,
right-continuous, adapted and integrable processes.\\

\noindent If ${\cal C}$ is a class of processes,
 we denote by ${\cal C}_0$ the set of processes $X$ with $X_0=0$ and
 by ${\cal C}_{loc}$
 the set  of processes $X$
 such that there exists a sequence of stopping times, $(T_n)_{n\geq 1}$,
 increasing stationarily to $T$ (i.e., $P(T_{n}= T)\rightarrow1$
as $n\rightarrow \infty$) and the stopped process $X^{T_n}$ belongs
to ${\cal C}$. We put $ {\cal C}_{0,loc}={\cal C}_0\cap{{\cal
C}}_{loc}$.\\

\noindent Very frequently,
throughout the paper, we will work with local martingale densities instead of
equivalent martingale measures. For this, we will use the following sets of densities
\begin{equation}\label{sigmadensities}
{\cal Z}^e_{loc}(S):=\Bigl\{Z={\cal E}(N)>0\ \Big|\ \ N\ \mbox{and}\ ZS\ \ \mbox{are local martingales}\Bigr\},\end{equation}
and/or
\begin{equation}\label{setofdernsities} {\cal
Z}^e_{q,loc}(S):=\Bigl\{Z={\cal E}(N)\in {\cal Z}^e_{loc}(S)\Big|\ \ \sum f_{q}(\Delta N)\in {\cal A}^+_{loc}\Bigr\},\end{equation}
where for any $r\in{\hbox{I\kern-.18em\hbox{R}}}$, the function $f_r$ is given by
\begin{equation}\label{def-fq}
f_{r}(x):=\left\{
            \begin{array}{llll}
              \displaystyle\frac{(1+x)^{r}-1-rx}{r(r-1)},\hskip 1cm  \hbox{if $r\not\in\{0,1\}$ and $x>-1$,} \\
              x-\log(1+x),\hskip 1.7cm \mbox{if $r=0$ and $x>-1$},\\
            (1+x)\log(1+x)-x,\hskip 0.55cm \mbox{if $r=1$ and $x\geq -1$,}\\
              +\infty,\hskip 3.3cm  \mbox{otherwise.}
            \end{array}
          \right.
\end{equation}

For a probability measure $R$ which is equivalent to $P$, the function $f_r$ together with the function $\Phi^R_r$ and the set ${\cal D}^+$ will play important roles in our analysis. Throughout the paper, $x^{tr} y$ will denote the inner product of $x$ and $y$ in $\mathbb R^d$. The function $\Phi^R_r$ ---also denoted by $\Phi^R_r=\Phi_r$ when $R=P$--- takes values in $(-\infty, +\infty]$ and is given by

\begin{equation}\label{Phip}
\Phi^R_r(\lambda):={{\lambda^{tr} b^R}\over{r-1}}+{1\over{2}}\lambda^{tr} c\lambda+\int f_r(\lambda^{tr} x)F^R(dx),\ \ \ \ \forall\ \lambda\in \mathbb R^d,\ \ \forall\ r\not=1,\end{equation}
while the set ${\cal D}$ is defined by
\begin{equation}\label{setD}
{\cal D}^+:=\left\{\theta\in \hbox{I\kern-.18em\hbox{R}}^d\ :\ \ 1+\theta^{tr} x>0,\ \ \ \ \ \ F(dx)-\mbox{almost all}\ x\in\hbox{I\kern-.18em\hbox{R}}^d\right\}.\end{equation}

\noindent The next technical assumption is crucial for the explicit forms of our main results. For some $p\in (-\infty,1)$ and  a positive local martingale $M={\cal E}(N)$, $N_0=0$, let ($\beta,f,g,N'$) be Jacod's parameters of $N$ (see Theorem \ref{representation} for details) and put the measure $F^M_t(dx):=(1+f_t(x))F_t(dx)$. Then, our main assumption that depends on $(p,M)$ is
\begin{assumption}\label{crucialassumption} For any predictable process $\lambda$ such that $\lambda\in \mathcal{D}^+$, $P\otimes A$-a.e., and every sequence of predictable processes, $(\lambda_n)_{n\geq 1}$, such that $\lambda_n\in int(\mathcal{D}^+)$,
and $\lambda_n\rightarrow \lambda$ $P\otimes A$-a.e., we have
\begin{equation}\label{assumptionA(p,R)}
\lim_{n\rightarrow +\infty}\int K_p(\lambda_n^{tr}x)F^M(dx)=\left\{
                                                                  \begin{array}{llll}
                                                                    +\infty, & \hbox{\quad on\quad $\Gamma$;} \\
                                                                    \\
                                                                    \displaystyle\int K_p(\lambda^{tr}x)F^M(dx), & \hbox{\quad on\quad$\Gamma^c$.}
                                                                  \end{array}
                                                                \right.\end{equation}
where $K_p(y):=y\left[1-(1+y)^{p-1}\right]$ and $\Gamma:=\{F^M(\mathbb R^d)>0$ and $\lambda\notin int(\mathcal{D}^+)\}$.
\end{assumption}

\begin{examples}\label{assA-example}
The {\bf Assumption \ref{crucialassumption}} is a technical condition essential for the solution of our stochastic optimization to belong to the interior of the effective domain of the functional that we minimize. Below, we cite some models that satisfy {\bf Assumption \ref{crucialassumption}}.\\
1) Assumption \ref{crucialassumption} is automatic when $S$ is continuous since, in this case, $F\equiv 0$ and $\Gamma=\emptyset$.\\
2) Assumption \ref{crucialassumption} is fulfilled for any model such that
 $$F_t(dx)=\sum_{i=1}^{n_t}\delta_{x_i(t)}(dx),\ \ \ \ \ \ P\otimes A-a.e.,$$ where $x_i(t)$ is a process with values in $\mathbb R^d$ and $n_t$ is a process with values in $\mathbb N$. Examples of models ---for which their kernels $F$ are atomic--- can be found in risk theory such as
 $X_t=\displaystyle\sum_{j=1}^{N_t}Y_j.$
 Here $(Y_i)_{i\geq 1}$ are iid and independent of the Poisson process $N$, and
 $
\displaystyle\sum_{k=1}^n P(Y_1=y_k)=1.$
Here $y_k\in \mathbb R^d$ and $n\in\mathbb N$, and $X_t$ represents the aggregate claims up to time $t$. 
\end{examples}

\noindent On the set $\Omega\times [0,\ T]$, we define two
$\sigma$-fields denoted by, ${{\cal O}}$ and ${{\cal P}}$, generated
by the adapted and RCLL processes and the adapted and continuous
processes respectively. On the set $\Omega\times [0,\
T]\times\mathbb R^d$, we consider the
$\sigma$-field ${\widetilde{\cal P}}={{\cal P}}\otimes{{\cal
B}}(\mathbb R^d)$ (resp. $\widetilde{{\cal
O}}={{\cal O}}\otimes{{\cal B}}(\mathbb R^d))$,
where ${\cal B}(\mathbb R^d)$ is the Borel $\sigma$-field for $\mathbb R^d$.\\
For any ${\widetilde{{\cal O}}}$-measurable functional, $g$, (hereafter
denoted by $g\in{\widetilde{{\cal O}}}$), we define $M^P_{\mu}(g\ |\
{\widetilde{{\cal P}}})$ to be the unique ${\widetilde{{\cal
P}}}$-measurable functional, when it exists,
 such that for any bounded $W\in{\widetilde{{\cal P}}}$,
$$
M^P_{\mu}(Wg):=E\Bigl(\int_0^T\int_{\mathbb R^d}
W( s,x)g( s,x)\mu( d s, d x)\Bigr) = M^P_{\mu}\Bigl(WM^P_{\mu} (g\
|\ {\widetilde{{\cal P}}})\Bigr).
$$

\noindent Throughout the paper, we will use frequently the following order to compare two stochastic processes.

\begin{definition}\label{order}
Let $X$ and $Y$ be two processes such that $X_0=Y_0$. Then, we write
$$ X\preceq Y$$ if $Y-X$ is a nondecreasing process. \end{definition}

\noindent This order was used to defined the minimal Hellinger martingale measures (or densities) in recent papers, see for instance Choulli and Stricker (2005, 2006) and Choulli et al. (2007). Below, we define a slight extension of the Hellinger process that incorporates the local change of probabilities.

\begin{definition}\label{hellinger-changemeasure}
(i) Let $Q$ be a probability measure and $Y$ be a $Q$-local
martingale such that $1+\Delta Y\geq 0$. Then, if the RCLL
nondecreasing process \begin{equation}\label{VEofY}
V^{(q)}(Y)={1\over{2}}\langle Y^c\rangle +\sum f_{q}(\Delta Y),
 \end{equation}
 is $Q$-locally integrable (i.e. $V^{(q)}(Y)\in {\cal A}^+_{loc}(Q)$), then its $Q$-compensator is called the Hellinger process of order $q$ of the local martingale $Y$ (or equivalently of ${{\cal E}}(Y)$) with respect to $Q$, and is denoted by $h^{(q)}(Y,Q)$ (respectively $h^{(q)}({{\cal E}}(Y),Q)$).\\
 (ii) Let $N\in{\cal M}_{0,\ loc}(P)$ such that $1+\Delta N> 0$ and $Y$ is a semimartingale such that $Y{{\cal E}}(N)$ is a $P$-local martingale and $1+\Delta Y\geq 0$. Then, if the process
 \begin{equation}\label{VeNNbar}
 {1\over{2}}\langle Y^c\rangle +\sum(1+\Delta N)f_{q}(\Delta Y),\end{equation}
 is $P$-locally integrable, then its $P$-compensator is called the Hellinger process of order $q$ of the local martingale ${{\cal E}}(Y)$ with respect to ${{\cal E}}(N)$, and is denoted by
 $h^{(q)}\left({{\cal E}}(Y),{{\cal E}}(N)\right)$.
\end{definition}

\begin{definition}\label{qHellinger}
Let $q$ be a real number. We call the minimal Hellinger martingale density of order $q$, the process $\widehat Z$ that belongs to ${\cal Z}_{q,loc}^e(S,Z)$ and satisfies
\begin{equation}\label{dominationHellinger}
h^{(q)}({\widehat Z},P)\preceq h^{(q)}({Z},P),\ \ \ \ \ \ \forall\ \ \ \ Z\in{\cal Z}_{q,loc}^e(S).\end{equation}\end{definition}

\noindent The following lemma proves that the concept of minimal Hellinger martingale density is stable under localization, and thus it is enough to describe this minimal martingale density locally. This lemma was first established in Choulli et al. (2011) for the case of exponential utility (i.e. $q=1$, the minimal entropy-Hellinger martingale density).

\begin{lemma}\label{LemmaforMEHM} Let $(T_n)_{n\geq 0}$ ($T_0=0$) be a sequence of stopping times that increases stationarily to $T$, and let $q\in\mathbb R$.
Suppose that for each $n$, $S^{T_n}$ admits the minimal Hellinger martingale density of order $q$ (called hereafter in the shorthand form of MHM density of order $q$), denoted by $\widetilde Z^{(n)}$. Then, $S$ admits the MHM density of order $q$, denoted by $\widetilde Z$ and is given by
$$
\widetilde Z:={{\cal E}}(\widetilde N)\ \ \ \ \mbox{and}\ \ \ \
\widetilde N:=\sum_{n\geq 1} I_{\Lbrack T_{n-1}, T_{n}\Lbrack}
{1\over{\widetilde Z^{(n)}_{-}}}\cdot \widetilde Z^{(n)}.$$
\end{lemma}
{\bf Proof.} The proof of this lemma is immediate and will be omitted.\qed
\subsection{Preliminaries on Random Field Utilities}
Throughout the paper, we call a random field utility, any ${\cal B}([0,T])\otimes{\cal B}(\mathbb R)\otimes{\cal F}$-measurable
 functional, $U(t,x,\omega)$ satisfying:\\
 (i) for any fixed $x$, the process $U(t,x,\omega)$ is a c\`adl\`ag and adapted process,\\
 (ii) and for any fixed $(t,\omega)$ the function $x\mapsto U(t,x,\omega)$ is strictly increasing and strictly concave.\\

 \noindent For a random field utility $U(t,x,\omega)$, a probability measure $Q$, a semimartingale $X$, and $x\in\mathbb R$
  such that $U(t,x,\omega)<+\infty $, we denote by
 \begin{equation}\label{admissibleset}
\displaystyle{\cal A}_{adm}(x, X, Q,U):=\Bigl\{\pi\in L(X)\ \Big|\
\sup_{\tau\in{{\cal T}_T}}E^Q\left[U\Bigl(\tau, x+(\pi\cdot
X)_{\tau}\Bigr)^-\right]<+\infty\Bigr\},\end{equation} the set of
admissible portfolios for the model $(x, X, Q, U)$. Here ${\cal
T}_T$ is the set of stopping times, $\tau$, such that $\tau\leq T$.
When $X=S$ and $Q=P$, for the sake of simply, we write ${\cal A}_{adm}(x,U)$.\\

\begin{definition}\label{forwardutilitiesdefi0} Consider a c\`adl\`ag semimartingale, $X$, and a probability measure, $Q$.
 Then, we call a forward utility for $(X,Q)$, any random field utility, $U:=(U(t, \omega, x))$, fulfilling the following self-generating property:\\
 a) The function $U(0,x)$ is strictly increasing and concave.\\
 b) For any $x\in (0,+\infty)$, there exists $\pi^*_x\in{\cal A}_{adm}(x,X,Q)$ such that
 $$
 U\Bigl(s,x+(\pi^*_x\cdot X)_s\Bigr)=E^Q\left[U\Bigl(t,x+(\pi^*_x\cdot X)_t\Bigr)|{\cal F}_s\right],\ \ \ \ \forall\ T\geq t\geq s\geq  0.$$
 c) For any $x\in (0,+\infty)$ and for any $\pi \in{\cal A}_{adm}(x,X,Q)$, we have
 $$
 U\Bigl(s,x+(\pi\cdot X)_s\Bigr)\geq E^Q\left[U\Bigl(t,x+(\pi\cdot X)_t\Bigr)|{\cal F}_s\right],\ \ \ \  \forall\ T\geq t\geq s\geq  0.$$
 When $X=S$ and $Q=P$, we simply call $U$ a forward utility.\end{definition}

 \begin{definition}\label{forwardutilitiesdefi}
 Let $X$ be a c\`adl\`ag semimartingale and $Q$ be a probability measure. Then, we call HARA forward utility for $(X,Q)$, any
 forward utility for $(X,Q)$ that takes one of the following forms 
 \begin{equation}\label{harafowards}
 U_p(t,x):=D(t)x^{p(t)},\ \ \ 
 U_0(t,x):={\widehat D}(t)\log\left(x\right)+\overline{D}(t),\ \ \ 
 U_e(t,x):=D_e(t)\exp\left(\gamma(t)x\right).\end{equation}
 Here $D=\left(D(t)\right)_{0\leq t\leq T}$, $p=\left(p(t)\right)_{0\leq t\leq T}$, $\gamma:=(\gamma(t))_{0\leq t\leq T}$, ${D}_e=({ D}_e(t))_{0\leq t\leq T}$, and  ${\widehat D}=({\widehat D}(t))_{0\leq t\leq T}$ and $\overline{D}=(\overline{D}(t))_{0\leq t\leq T}$ are stochastic processes.
\end{definition}

\noindent In this paper, we focus on the first and the second types of HARA utilities ( i.e. the two cases of $U_p$ and $U_0$). The exponential-type forward utilities (i.e. the case of $U_e$) is completely analyzed in Choulli et al. (2011). In our view, it is impossible to describe all cases in one single form. This fact is highly supported by the different forms that we found for the risk-aversion processes $p$ and $\gamma$.

\begin{proposition}\label{PropertiesOfForward}
Let $U:=U(t,\omega,x)$ be a random field utility and $S$ be a semimartingale. Then the following hold.\\
(i) If $U$ is a forward utility for $(S,P)$, then for any stopping time $\tau\in\mathcal{T}_T$, the functional
\begin{equation}\label{Ubar} \overline U(t,\omega,x):=U(t\wedge
\tau(\omega),\omega,x),\end{equation}
is a forward dynamic utility for $(S^{\tau},P)$.\\
(ii) Consider a probability measure $Q$ that is absolutely
continuous with respect to $P$ with the density process denoted by
$Z$. Then, the random field utility \begin{equation}\label{UQ}
U^Q(t,\omega,x):=U(t,\omega,x)Z_t(\omega),\end{equation} is a
forward utility for $(S,P)$ if and only if $U$
is forward utility for $(S,Q)$.
\end{proposition}

{\bf Proof.} The proof of this proposition is easy and can be found in Choulli et al. (2011). \qed
%
%
%





\section{Parametrization of the Power-Type Forward Utilities}\label{powercase}
In this section, we will parameterize the forward utilities of the following form
\begin{equation}\label{harapower}
 U_p(t,x):=D(t)x^{p(t)}.
\end{equation}
To this end we start by deriving some useful properties of the portfolio rate process defined in the following.
\begin{definition}\label{portfolioratedefi}
 Let $\pi$ be a portfolio (i.e. $\pi\in L(S)$) and $x>0$ such that
\begin{equation}\label{existenceofportfoliorate}
x+\pi\cdot S>0,\ \ \mbox{and}\ \ x+(\pi\cdot S)_{-}>0.\end{equation}
Then, we define the {\it portfolio rate} by
\begin{equation}\label{portfoliorate}
\theta_x:=\Bigl(x+(\pi\cdot S)_{-}\Bigr)^{-1}\pi.\end{equation}
\end{definition}

\begin{lemma} Let $\pi$ be a portfolio and $x>0$ such that (\ref{existenceofportfoliorate}) holds. Then, the following hold.\\
(i) The portfolio rate $\theta_x$ is $S$-integrable and $ {\cal E}(\theta_x\cdot S)=\left(x+\pi\cdot S\right)/x>0.$\\
(ii) There is a one-to-one correspondence between $\pi$ and its portfolio rate $\theta_x$ via (\ref{portfoliorate}) and
\begin{equation}\label{reverse}
\pi=x{\cal E}_{-}(\theta_x\cdot S)\theta_x.\end{equation}
\end{lemma}
{\bf Proof.} The proof of this lemma is obvious and will be omitted.\qed\\

For our analysis when dealing with utilities $U_p$ ( in both cases of $p=0$ and $p\not=0$), it is more convenient to deal with the portfolio rate than the portfolio itself. Therefore we need to define the set of admissible portfolio rates that we will use throughout the paper. For any semimartingale $X$, any probability $Q$, any random utility $U$, and any initial capital $x>0$, we associate the set of admissible portfolio rates $\Theta(x,X,U,Q)$ defined by
\begin{equation}\label{porfoliorateadmissible}
\Theta(x,X,Q,U):=\Bigl\{\theta\in L(X)\ \big|\ \ \ {\cal E}(\theta\cdot X)>0\ \ \&\ \ \pi_x:=x{\cal E}_{-}(\theta\cdot X)\theta\in \mathcal{A}_{adm}(x,X,Q,U)\Bigr\}.\end{equation}
Again, when $X=S$ and $Q=P$, the set of admissible portfolio rates will be denoted by $\Theta(x,U)$ for the sake of simplicity.

\begin{proposition}\label{optimalportfolio} Suppose that $S$ is locally bounded and ${\cal Z}_{loc}^e(S)\not=\emptyset$. Let $p\in (-\infty,1)$ and consider
\begin{equation}\label{utility}
{\overline U}_p(t,x):=\left\{\begin{array}{llll}D(t)x^p,\ \ \ \mbox{if}\ \ p\not=0\\
\\
\overline{D}(t)+{\widehat D}(t)\log(x),\ \ \ \mbox{if}\ \ p=0.\end{array}\right.\end{equation}
For any $x\in (0,+\infty)$, consider the following maximization problem
\begin{equation}\label{optimization1}
\max_{\pi\in{\cal A}_{adm}(x)}E{\overline U}_p\left(T,x+(\pi\cdot S)_T\right),\end{equation}
where the set ${\cal A}_{adm}(x, {\overline U}_p)$ is defined in (\ref{admissibleset}). Then following assertions hold.\\
(1) For any $x\in (0,+\infty)$, if the solution to (\ref{optimization1}) ---that we denote by $\widetilde\pi_x$--- exists, then
\begin{equation}\label{positivity}
x+\widetilde\pi_x\cdot S>0,\ \ \ \ \ \ \mbox{and}\ \ \ x+(\widetilde\pi_x\cdot S)_{-}>0.\end{equation}
(2) The optimal portfolio rate for ${\overline U}_p$ with initial capital $x$, that we denote by $\widetilde\theta_x:=\left(x+(\widetilde\pi_x\cdot S)_{-}\right)^{-1}\widetilde\pi_x$,
is independent of $x\in (0,+\infty)$ (or equivalently $\widetilde\pi_x/x$ is independent of $x$).
\end{proposition}
{\bf Proof.} It is clear from Kramkov and Schachermayer (1999), that the random variable $x+(\widetilde\pi_x\cdot S)_T$ is positive, and the process $(x+\widetilde\pi_x\cdot S)Z$ is a supermartingale, for any $Z\in {\cal Z}_{loc}^e(S)\not=\emptyset$. This implies that both processes $x+\widetilde\pi_x\cdot S$ and $x+(\widetilde\pi_x\cdot S)_{-}$ are positive and assertion (1) follows. To prove assertion (2), it is enough to remark that for any $x\in (0,+\infty)$, $x\widetilde\pi_1\in {\cal A}_{adm}(x, {\overline U}_p)$, and for any $\pi\in {\cal A}_{adm}(x, {\overline U}_p)$, we have $x^{-1}\pi\in {\cal A}_{adm}(1, {\overline U}_p)$. This ends the proof of the proposition.
\qed\\

The remaining part of this section contains three subsections. The first subsection (Subsection \ref{subsection1forpower}) deals with the description of the process $p$, while the second subsection (Subsection \ref{subsection:proofoftheorem6}) focuses on the process $D$. In the last subsection (Subsection \ref{sectyionproofofTheorem4}), we detail the proof of Theorem \ref{case-pre+var-Hellinger} which represents the back-bone for the general result of the second subsection. Throughout the rest of the paper we denote by $U_p(t,x)$ the functional defined in (\ref{harapower}). This functional depends on the uncertainty $\omega\in\Omega$ also, while throughout the paper we will use the shorthand $U_p(t,x)$ for the sake of simplicity.





\subsection{The Dynamic of the Risk-Aversion Process $p$}\label{subsection1forpower}

This subsection constitutes our first major step within our parametrization of the functional defined in (\ref{harapower}). Below, we state the principal result of this subsection.
\begin{theorem}\label{pconstant}
Suppose that $S$ is locally bounded and ${\cal Z}_{loc}^e(S)\not=\emptyset$. Let $U_p(t,x)$ be defined in (\ref{harapower}) such that
\begin{equation}\label{null-strategy}
\sup_{\tau\in\mathcal{T}_{T}}E\left[D(\tau)^{-}\right]<+\infty,\ \ \mbox{and}\ \ p=\left(p(t)\right)_{0\leq t\leq T}\ \mbox{is locally bounded}.\end{equation}
 If $U_p(t,x)$ is a forward utility, then $p(\omega,t)=p(0)$ $P-$almost all $\omega\in\Omega$, and for all $t\in[0,T].$
\end{theorem}

The proof of Theorem \ref{pconstant} requires two intermediary steps that will be detailed in Lemmas \ref{lem-p-positive} and \ref{Cor-continuous-p}. These two lemmas prove that Theorem \ref{pconstant} holds when $p$ has a constant sign.


\begin{lemma}\label{lem-p-positive}
Suppose that $p=(p(t))_{t\geq 0}$ is positive, (\ref{null-strategy}) holds and $U_p$ is a forward utility. Then, the process $p$ is constant in $(\omega,t)$ (i.e. Theorem \ref{pconstant} holds true in this case).
\end{lemma}

{\bf Proof.}
Since $U_p(t,x)$ is a random field utility, then it is strictly increasing in the variable $x$ for any $(t,\omega)\in [0,T]\times\Omega$. This implies that $pD>0$ and hence $D>0$ (since $p$ is positive). By stopping and using Proposition \ref{PropertiesOfForward}--(i), we can assume ---without loss of generality--- that $p$ is bounded. Thus, it is easy to see that in this case the null portfolio rate $\theta=0$ is admissible for any $x>0$ --- i.e. $0\in\Theta(x, U_p)$--- due to (\ref{null-strategy}). Therefore, the optional sampling theorem and the supermartingale property of  $U_p(t,x)$ lead to
\begin{equation}\label{power-ppositive-super}
E\left(D(\sigma)x^{p(\sigma)}|\mathcal{F}_{\tau}\right)\leq D(\tau)x^{p(\tau)},
\end{equation}
for any $x>0$ and any stopping times $\sigma$ and $\tau$ satisfying $\sigma\geq \tau$. By putting
$$Q:=\frac{D(\sigma)/D(\tau)}{E(D(\sigma)/D(\tau))}\cdot P,\ \ \ \ \ \ \ \mbox{and}\ \ \ \Delta:=p(\sigma)-p(\tau),$$
we conclude that (\ref{power-ppositive-super}) becomes
$$
E^{Q}\left(e^{\log(x)\Delta}-1|\mathcal{F}_{\tau}\right)\leq  C_{Q}:=\Bigl(E(D(\sigma)|\mathcal{F}_{\tau})\Bigr)^{-1}-1,\ \ \ \ \ \mbox{for all}\,\,x>0.
$$
Due to $e^{\log(x)\Delta}-1=e^{\log(x)\Delta^+}+e^{-\log(x)\Delta^-}-2,$ the above inequality is equivalent to
\begin{equation}\label{power-Delta+-}
E^{Q}\left(e^{\log(x)\Delta^+}+e^{-\log(x)\Delta^-}|\mathcal{F}_{\tau}\right)\leq C_{Q}+2,\ \ \ \ \ \mbox{for all}\,\,x>0.
\end{equation}
Thanks to Jensen's inequality, (\ref{power-Delta+-}) yields
$$
\exp\left(\log(x)E^{Q}(\Delta^+|\mathcal{F}_{\tau})\right)+\exp\left(-\log(x)E^{Q}(\Delta^-|\mathcal{F}_{\tau})\right)\leq C_Q+2.
$$
This inequality holds for all $x>0$ if and only if
$$
E^{Q}(\Delta^+|\mathcal{F}_{\tau})=E^{Q}(\Delta^-|\mathcal{F}_{\tau})=0,\quad P-a.s.
$$
or equivalently $p(\sigma)=p(\tau)$ $P-a.s.$ Since the pair of stopping times is arbitrary, then the proof of the lemma follows immediately.\qed\\


 The positivity condition on $p$ is crucial in our proof of Lemma \ref{lem-p-positive}, and without it this approach will not be conclusive at some stage. Thus, for the case of negative $p$, we will proceed differently by using a result of Berrier et al. (2009), where the authors tried to measure the effect of the forward property of a random field utility on its random field conjugate that can be defined as follows. For any random field utility, $U(t,x)$, $x\in\mathbb R^+$ we define its Frenchel-Legendre conjugate (called hereafter random field conjugate),
$V(\omega,t,y)$, by
$$
V(t,\omega,y):=\sup_{x>0}\Bigl(U(t,\omega,x)-xy\Bigr),\ \ \ \ \mbox{for}\ \  t\geq0,\ \ \ \ \ y> 0.$$
The random field conjugate of $U_p$ defined in (\ref{harapower}) is given by
\begin{equation}\label{Vpdualconjugate}
V_p(t,y)=-(D(t)p(t))^{1-q(t)}y^{q(t)}\left(q(t)\right)^{-1},
\end{equation}
where $q$ is the conjugate process of $p$ and is given by $q(t):={{p(t)}\over{p(t)-1}}$.

\begin{lemma}\label{Cor-continuous-p} Suppose that the process $p=(p(t))_{0\leq t\leq T}$ is negative, (\ref{null-strategy}) holds and $U_p$ is a forward utility.
Then, $p$ is constant in $(t,\omega)$, i.e. $p(t)=p(0)$, $P-a.s..$
\end{lemma}

{\bf Proof.} Consider $t\geq 0$, arbitrary but fixed. A direct application of Proposition \ref{prop-dual-submartingale} (see at the end of this proof)  to $U_p(t,x)=D(t)x^{p(t)}$ and its random field conjugate defined in (\ref{Vpdualconjugate}) implies that for any $T'\in[t,+\infty)$, any $Z\in\mathcal{Z}^{e}_{loc}(S)$,  and $\eta\in L^{0}_{+}(\mathcal{F}_{t})$, we have
\begin{equation}\label{case2-applyProp}
E\left(\frac{(D(T')p(T'))^{1-q(T')}}{q(T')}(\eta\frac{Z_{T'}}{Z_{t}})^{q(T')}|\mathcal{F}_{t}\right)\leq
\frac{(D(t)p(t))^{1-q(t)}}{q(t)}\eta^{q(t)}.
\end{equation}
By choosing $\eta=Z_te^{\alpha}$, and putting $X_{s}:=\frac{(D(s)p(s))^{1-q(s)}}{q(s)}Z_{s}^{q(s)}>0$, the equation (\ref{case2-applyProp}) becomes
\begin{equation}\label{case2-Xintro}
E\left(X_{T'}e^{\alpha(q(T')-q(t))}|\mathcal{F}_{t}\right)\leq X_{t}.
\end{equation}
Since $X$ is positive (which is due to $q(t)={{p(t)}\over{p(t)-1}}>0$), then we derive
\begin{equation}\label{proof-thm-max}
\max\Big
\{e^{\varepsilon\alpha^+}E\left(X_{T'}I_{\{q(T')-q(t)\geq \varepsilon\}}|\mathcal{F}_{t}\right),e^{\varepsilon\alpha^-}E\left(X_{T'}I_{\{q(T')-q(t)\leq -\varepsilon\}}|\mathcal{F}_{t}\right)\Big\}\leq X_{t},
\end{equation}
for any $\alpha\in \RR$ and any $\varepsilon>0$, where $\alpha=\alpha^+-\alpha^-$. Therefore, we deduce that
$$q(T')=q(t),\quad P-a.s.\quad \forall\ \ \ T\geq T'\geq t\geq 0.$$
This ends the proof of the lemma. \qed\\

The proof of the previous lemma is essentially based on the following.
\begin{proposition}\label{prop-dual-submartingale}
If $U(t,x)$ is a forward utility, then for any $T'$, $0\leq t\leq T'\leq T$ and any $\eta\in L^{0}_{+}(\mathcal{F}_{t})$, we have
\begin{equation}\label{dual-submartingale}
V(t,\eta)\leq \underset{Z\in\mathcal{Z}^{e}_{loc}(S)}{\mathrm{ess\,inf}}\,\,
E\left(V(T',\eta\frac{Z_{T'}}{Z_{t}})\Big|\mathcal{F}_{t}\right),\quad P-a.s..
\end{equation}
\end{proposition}

\begin{proof} The proof of this proposition can be found in Berrier et al. (2009). \end{proof}

The remaining part of this subsection is devoted to the proof of Theorem  \ref{pconstant}.\\

 {\bf Proof of Theorem \ref{pconstant}:} If the process $p$  is either positive or negative, then the proof of this theorem follows from Lemmas \ref{lem-p-positive} and \ref{Cor-continuous-p} respectively. Thus, the proof of this theorem will be achieved once we prove that the process $p$ has a constant sign (i.e. $p(t)p(0)>0,\ P-a.s.$ for all $t\in [0,T]$). To this end, we assume that
 \begin{equation}\label{assumtiononp}
 \sup_{0\leq t\leq T}\vert p(t)\vert>0,\ \ \  P-a.s.,\end{equation}

and consider the following stopping time
$$
\tau:=\inf\left\{t\geq 0\ \Big|\ p(t)p(0)<0\right\}\wedge T.$$
Due to (\ref{null-strategy}), the null portfolio belongs to ${\cal A}_{adm}(x,U_p)$ for any $x\geq 1$. Hence, $D$ and $U_p(t,e)$ are two c\`adl\`ag supermartingales that never vanish, and thus $p(t)=\log\left(U_p(t,e)/D(t)\right)$ is a right continuous and adapted process. By combining this right continuity of $p$ and  (\ref{assumtiononp}), we deduce that $p$ has a constant sign if and only if
 \begin{equation}\label{constantsignforp}
  p(0)p(\tau)>0,\ \ \ \ \ \ P-.a.s.\end{equation}
 Since $p$ is locally bounded, there is no loss of generality in assuming $p$ bounded. To prove (\ref{constantsignforp}), we will proceed by distinguishing whether $p(0)<0$ or $p(0)>0$ in parts a) and b) respectively.\\

\noindent {\bf a)}
Suppose that $p(0)<0$, and hence $D(0)<0$. Then, due to (\ref{null-strategy}), we deduce that the null strategy is admissible for any $x>0$. Hence, $D(t)x^{p(t)}$ is a supermartingale, for any $x>0$, and we have
$$
E\left(D(\tau)x^{p(\tau)}I_{\{p(\tau)>0\}}\ \Big|\ {\cal F}_0\right)\leq -E\left(D(\tau)x^{p(\tau)} I_{\{p(\tau)<0\}}\ \Big|\ {\cal F}_0\right)+D(0)x^{p(0)}.$$
By letting $x$ goes to infinity and using Fatou's lemma we conclude that we should have $P\left(p(\tau)>0\right)=0$ (otherwise we will have a contradiction from the above inequality). This proves (\ref{constantsignforp}).\\

\noindent {\bf b)} Suppose that $p(0)>0$, or equivalently $D(0)>0$. Then for any $n\geq 1$, there exists $\theta_n\in {\cal A}_{adm}(n^{-1})$ such that
\begin{equation}\label{power-p-opt}
E\left\{D(\tau)n^{-p(\tau)}(1+(\theta_n\cdot S)_{\tau})^{p(\tau)}\right\}=D(0)n^{-p(0)}.
\end{equation}
Thanks to Lemma A1.1 in Delbaen and Schachermayer (1994) , there exists a sequence of non-negative real numbers, $(\alpha_k)_{k=n,...,N_n}$, such that
$$
\sum_{k=n}^{N_n}\alpha_k=1 \ \ \mbox{and}\ \  Y_n:=1+\sum_{k=n}^{N_n} \alpha_k(\theta_k\cdot S)_{\tau}\ \ \ \mbox{converges almost surely to}\ \ Y\geq 0.$$
Thanks to Fatou's lemma and $\mathcal{Z}^e_{loc}(S)\not=\emptyset$, we obtain
$$E\left(ZY\right)\leq \lim_{n\rightarrow+\infty}E\left(ZY_n\right)\leq 1,$$
for some $Z\in \mathcal{Z}^e_{loc}(S)$. This implies that $0\leq Y<+\infty\ \ P-a.s.$
Consider $(X_n)_{n\geq 1}$ given by
$$X_n:=D(\tau)n^{-p(\tau)}Y_n^{p(\tau)}-D(\tau)Y_n^{p(\tau)}.$$
It is easy to check that $
X_n\leq 0,\ \ \ P-a.s.$, and by distinguishing the cases of $\{p(\tau)>0\}$ and $\{p(\tau)<0\}$ separately, on the one hand we derive
\begin{equation}\label{power-xn-lim}
\lim_{n\rightarrow +\infty}X_n=\left\{
                                   \begin{array}{ll}
                                     -D(\tau)Y^{p(\tau)}, &\quad  \hbox{if $\ p(\tau)>0$;} \\
                                     -\infty, & \quad \hbox{if $\ p(\tau)<0$.}
                                   \end{array}
                                 \right.
\end{equation}
On the other hand, we have
\begin{eqnarray}
X_n&\geq& \sum_{k=n}^{N_n}\alpha_kn^{-p(\tau)}D(\tau)(1+\theta_k\cdot S_{\tau})^{p(\tau)}-D(\tau)Y_n^{p(\tau)}\notag\\
&\geq&\sum_{k=n}^{N_n}\alpha_kk^{-p(\tau)}D(\tau)(1+\theta_k\cdot S_{\tau})^{p(\tau)}-D(\tau)Y_n^{p(\tau)}.\label{power-xn-geq}
\end{eqnarray}
Then, by taking expectation in both sides of (\ref{power-xn-geq}), and using (\ref{power-p-opt}) and the supermartingale property of $U_p(\tau,1+ \sum_{k=n}^{N_n}\alpha_k\theta_k\cdot S_{\tau})$, we get
\begin{equation}\label{power-xn-E}
E(X_n)\geq D(0)\left[\sum_{k=n}^{N_n}\alpha_kk^{-p(0)}-1\right].
\end{equation}
Since $X_n$ is nonpositive, again Fatou's lemma to the left-hand-side term of (\ref{power-xn-E}) leads to
\begin{equation}\label{power-fatou-xn}
E(\lim_{n\rightarrow +\infty}X_n)\geq -D(0)>-\infty.
\end{equation}
Then, thanks to (\ref{power-xn-lim}) and (\ref{power-fatou-xn}), we deduce that $P(p(\tau)<0)=0$. Hence (\ref{constantsignforp}) holds, and the theorem is proved under the assumption (\ref{assumtiononp}). The proof of the theorem will be completed if we prove that this assumption actually holds. To this end, consider
$$
\tau_0:=\inf\{t\geq 0\ \ |\ p(t-)=0\}\wedge T.$$

Since $p$ never vanishes (since $U_p$ is a random field utility), we deduce that on $\{\tau_0<T\}$, we have $p(\tau_0-)=0$ $P-a.s$. This implies that $\tau_0$ is a predictable stopping time that is announced by a sequence of stopping times $(\sigma_n)_{n\geq 1}$ satisfying
$$
\sup_{0\leq t\leq\sigma_n}\vert p(t)\vert>0\ \ P-a.s.$$
Thus, $p^{\sigma_n}$ fulfills the assumption (\ref{assumtiononp}) and hence it is constant equal to $p(0)$. Then, on $\{\tau_0<T\}$ we have  $0=p(\tau_0-)=p(0)\not=0$, which implies that $\tau_0=T\ \mbox{ and}\ p(T-)=p(0)\not =0\  P-a.s.$. This proves that both processes $p(t)$ and $p(t-)$ never vanish, and (\ref{assumtiononp}) follows immediately. This ends the proof of the theorem.\qed

\vskip 0.5cm

\subsection{The Dynamic of the Process $D$}\label{subsection:proofoftheorem6}

\noindent In this subsection, we develop our second and last step of our parametrization for the power-type forward utilities. Thanks to Theorem \ref{pconstant}, we will assume ---throughout the rest of this section--- that the process $p$ is constant in $(\omega, t)$, and we will describe $\left(D(t)\right)_{0\leq t\leq T}$ as well as the optimal portfolio for the utility maximization problem associated to $U_p(t,x)$. These results will be presented in two theorems that are stated in the increasing order of generality. First, we describe the process $D$ that is predictable with finite variation (Theorem \ref{case-pre+var-Hellinger}). Afterwards, we drop the predictability and the finite variation assumptions, and determine the general form of $D$ (Theorem \ref{power-characterization-ass}).\\

\begin{theorem}\label{case-pre+var-Hellinger}
Let $p$ be a real number such that $0\not= p<1$, $q$ is its conjugate ($q:=\frac{p}{p-1}$), and the set ${\cal D}^+$ is given by (\ref{setD}). Suppose that $D(t)$ is a c\`adl\`ag and predictable process with finite variation, $S$ is locally bounded, ${\cal Z}_{loc}^e(S)\not=\emptyset$, and Assumptions \ref{crucialassumption} with $M\equiv 1$ holds .\\
Then, the assertions (1) and (2) below are equivalent.\\
(1) $U(t,x)=D(t)x^{p}$ is a forward utility with the optimal portfolio rate $\widehat{\theta}$.\\
(2) The minimal Hellinger martingale density of order $q$, $\widetilde{Z}$, exists and satisfies:  \\
(2.a) The process $\widehat{Z}:=\widetilde{Z}{\cal E}\left(\widehat{\theta}\cdot S\right)$ is a martingale.\\
(2.b) The process $D$ is given by
\begin{equation}\label{pre+var-characterD(t)}
D=D_{0}{\cal E}\left(q(q-1)h^{(q)}(\widetilde{Z},P)\right)^{p-1}.
\end{equation}
(2.c) $P\otimes A-$almost all $(\omega, t)$, the optimal portfolio rate, $\widehat{\theta}$, belongs to int$\left({\cal D}^+\right)$, and is a root for
\begin{equation}\label{eq-root-thetatilde}
b+(p-1)c\theta+\int\left[(1+\theta^{tr}x)^{p-1}-1\right]xF(dx)=0,\ \ \ \ P\otimes A-a.e.
\end{equation}
\end{theorem}

\noindent The proof of this theorem requires a number of intermediate results that are interesting in themselves. Technically, Theorem \ref{case-pre+var-Hellinger} is the back-bone of this subsection. Thus, for the sake of clear exposition, we will postpone its proof to Subsection \ref{sectyionproofofTheorem4}. In the following, we will highlight the importance of Theorem \ref{case-pre+var-Hellinger} on the particular case when $S$ is continuous, and afterwards we will deal with describing $D$ in the general case.

\begin{corollary}\label{corollaryPredictbaleD} Suppose that $D(t)$ is a c\`adl\`ag and predictable process with finite variation, $S$ is continuous and ${\cal Z}_{loc}^e(S)\not=\emptyset$. Then the following are equivalent.\\
(1) $U(t,x)=D(t)x^{p}$ is a forward utility with the optimal portfolio rate $\widehat{\theta}$.\\
(2) The optimal portfolio rate $\widehat{\theta}$ is a root for
$b+(p-1)c\theta=0,\ P\otimes A-a.e,$
and the following properties hold:\\
(2.a) The process $D$ is given by
$D_t=\displaystyle D_{0}\exp\left({{q}\over{2}}\int_0^t\widehat\theta_u^{tr} c_u\widehat\theta_u dA_u\right),\ \ \ \ 0\leq t\leq T.
$\\
(2.b) The process $\widehat{Z}:={\cal E}\left((p-1)\widehat\theta\cdot M^S\right){\cal E}\left(\widehat{\theta}\cdot S\right)$ is a martingale, where $M^S$ is the local martingale part of $S$.\\
\end{corollary}

{\bf Proof.} The proof of this corollary is straightforward from Theorem \ref{case-pre+var-Hellinger} and from the fact that, when $S$ is continuous, the minimal Hellinger densities of all order $r$ coincide with the minimal martingale density (see Choulli et al. (2007)), and Assumption \ref{crucialassumption} holds. This ends the proof of this corollary.\qed


\begin{remark} Theorems \ref{pconstant} and \ref{case-pre+var-Hellinger} claim that the set of forward utilities having the form of (\ref{harapower}) with predictable and finite variation $D$ is parameterized by the constant $p(0)$ which is the initial value of the risk-aversion process. Given $p(0)$, the process $D$ is uniquely and explicitly determined using the Hellinger process of an optimal pricing density. This density is the minimal Hellinger martingale density that is calculated explicitly in Choulli et al. (2007) (see also Choulli and Stricker (2005, 2006) for other cases of $p$) and plays --herein-- its natural r\^ole of the ``optimal dual process'' for the power-type forward utility maximization.  Furthermore, the optimal portfolio rate is explicitly given via a point-wise $\mathbb R^d$-equation (\ref{eq-root-thetatilde}).\\
\end{remark}

\noindent  Now, we are in the stage of announcing our description of the general form of $D$ and achieve the parametrization of $U_p$ --defined in (\ref{harapower})-- in its full generality.

\begin{theorem}\label{power-characterization-ass}
Consider ${U}_p$ given by (\ref{harapower}), where $p$ is assumed to be constant in $(-\infty,0)\cup(0,1)$ and $D$ satisfies
 \begin{equation}\label{Ddecomposition}
D(t)=D(0){\cal E}(N^D)\exp(a^D).\end{equation}
Here $N^D$ is a local martingale and $a^D$  is a predictable process with finite variation such that $Z^D:={\cal E}(N^D)>0$. Suppose that $S$ is locally bounded, ${\cal Z}_{loc}^e(S)\not=\emptyset$, and Assumptions \ref{crucialassumption} with $M\equiv Z^D$ holds. Then, the assertions (1) and (2) below are equivalent. \\
(1) ${U}_p$ is a forward utility with the optimal portfolio rate $\widehat\theta$.\\
(2) The following properties hold.\\
(2.a) The minimal Hellinger martingale density of order $q$ with respect to $Z^{D}$, denoted by $\widetilde{Z}^{D}$, exists and
\begin{equation}\label{characterD(t)-power}
D=D_{0}Z^{D}{\cal E}\left(q(q-1)h^{(q)}(\widetilde{Z}^{D},Z^{D})\right)^{p-1}.
\end{equation}
(2.b) The optimal portfolio rate, $\widehat{\theta}$, belongs to $int({\cal D}^+)$, and is a root for
\begin{equation}\label{root-theta-general0}
{ b^{D}+(p-1)c\theta+\int\left[(1+\theta^{T}x)^{p-1}-1\right]xF^{D}(dx)=0},\ \ \ \ P\otimes A-a.e.
\end{equation}
(2.c) The process $\widehat{Z}:=Z^{D}\widetilde{Z}^{D}{\cal E}(\widehat{\theta}\cdot S)$ is a martingale.\\
Here, $b^D$ and $F^D$ are given by \begin{equation}\label{jacodparametersunderD}
\displaystyle{ b^{D}:=b+c\beta+\int f(x)xF(dx),\quad \mbox{and}\ \  F^{D}(dx):=(1+f(x))F(dx),}
\end{equation}
 and $\left(\beta, f, g, \overline{N}^{D}\right)$ are the Jacod's parameters for $N^{D}$ guaranteed by Theorem \ref{representation}.
\end{theorem}

{\bf Proof.} We start by proving $(1)\Rightarrow (2)$. Suppose that assertion (1) holds. 
\noindent Let $(T_{n})_{n\geq 1}$ be a sequence of stopping times that increases stationarily to $T$ such that $(Z^{D})^{T_{n}}$ is a martingale. Put  $Q_{n}:=Z^{D}_{T_{n}}\cdot P$. Then, due to Lemma \ref{PropertiesOfForward}, we conclude that $U_{n}(t,\omega,x):=D_{0}\exp(a^{D}_{t\wedge T_{n}})x^{p}$ is a forward dynamic utility for $(S^{T_{n}},Q_{n})$ with the optimal portfolio rate $\widehat\theta_n:=\widehat\theta I_{\Rbrack 0,T_n\Lbrack}$. Hence, a direct application of Theorem \ref{case-pre+var-Hellinger} to $(S^{T_{n}},Q_{n},U_{n}, \widehat\theta_n)$ implies the existence of the minimal Hellinger martingale density for this model, denoted by $\widetilde Z^{D,n}$, that satisfies
\begin{equation}\label{Unequations}
\begin{array}{lll}\exp(a^{D}_{t\wedge T_{n}})={\cal E}_{t\wedge T_{n}}\left(q(q-1)h^{(q)}(\widetilde{Z}^{D,n},Q_{n})\right)^{1/(q-1)},\ \ \ \ \ \ 0\leq t\leq T,
\end{array}\end{equation}
and on $\Rbrack 0,T_n\Lbrack$, $\widehat\theta$ belongs to int$({\cal D}^+)$ and is a root for (\ref{root-theta-general0}). Thus, it is clear that, this last statement implies assertion (2.b). By virtue of Lemma \ref{LemmaforMEHM}, we deduce that the minimal Hellinger martingale density of order $q$ with respect to $Z^D$ (denoted by $\widetilde Z^D$) exists and
 $$
 h^{(q)}_t(\widetilde{Z}^{D,n},Q_{n})=h^{(q)}_{t\wedge T_n}(\widetilde{Z}^{D},Z^D).$$
 Therefore, a combination of this equality with (\ref{Unequations}) leads to the assertion (2.a).\\

\noindent Due to Proposition \ref{lem-represofMHE-changemeasure} (see formula (\ref{ZtildeZ}) and notice that our $\widehat\theta$ here is a version of $\widetilde\beta$ of that proposition) and (\ref{characterD(t)-power}), we derive
\begin{equation}\label{ZhatUtility}\begin{array}{lll}
\widehat Z=Z^{D}\widetilde{Z}^{D}{\cal E}({\widehat{\theta}}\cdot S)&=&Z^{D}{\cal E}\left(\widetilde{H}^{D}\cdot S+q(q-1)h^{(q)}(\widetilde{Z}^{D},Z^{D})\right)^{p-1}{\cal E}(\widehat{\theta}\cdot S)\\
&=&Z^{D}{\cal E}(\widehat{\theta}\cdot S)^{p-1}{\cal E}\left(q(q-1)h^{(q)}(\widetilde{Z}^{D},Z^{D})\right)^{p-1}{\cal E}(\widehat{\theta}\cdot S)\\
&=&\left(D_{0}x^{p}\right)^{-1}U(t,x{\cal E}_{t}(\widehat{\theta}\cdot S)).
\end{array}\end{equation}
This proves that $\widehat Z$ is a martingale, since $U(t,x)$ is a forward utility with optimal portfolio rate $\widehat\theta$. This ends the proof of $(1)\Rightarrow (2)$.\\

\noindent In the remaining part of this proof, we will address $(2)\Rightarrow (1)$. Suppose that assertion (2) is fulfilled, and remark that (\ref{ZhatUtility}) remains valid as long as assertion (2--a) holds. Thus, we obtain
\begin{eqnarray*}\label{cont-thetatilde-Utrue}
{U}_p\left(\cdot,x{\cal E}\left(\widehat{\theta}\cdot S\right)\right)&=&D_{0}x^{p}\widehat{Z},
\end{eqnarray*}
and due to assertion (2-c), we conclude that $U\left(\cdot,x{\cal E}\left(\widehat{\theta}\cdot S\right)\right)$ is a martingale for any $x>0$. Furthermore, for any admissible strategy $\theta$, we have
\begin{eqnarray}
{U}_p\Bigl(t,x{\cal E}_t\left({\theta}\cdot S\right)\Bigr)&=&D_{0}x^{p}\widehat{Z}_t{\cal E}_t\left({\theta}\cdot S\right)^p{\cal E}_t\left(\widehat{\theta}
\cdot S\right)^{-p}.\label{thetatilde-Usuper-changemeasure}
\end{eqnarray}
Thanks to $pD_0>0$ ($U(t,x)$ is a random field utility), Proposition \ref{theta-supermatingale} (take $\widetilde Z:=Z^D \widetilde Z^D$ which is a martingale density for $S$ by definition  of $\widetilde Z^D$), and
$$\sup_{\tau\in\mathcal{T}_{T}} E^{\widehat{Q}}\left\{{\cal E}_{\tau}\left({\theta}\cdot S\right)^p{\cal E}_{\tau}\left(\widehat{\theta}
\cdot S\right)^{-p}\right\}=-\frac{1}{D_{0}x^{p}}\sup_{\tau\in\mathcal{T}_{T}} E\Bigl[U\Bigl(\tau,x{\cal E}_{\tau}\left({\theta}\cdot S\right)\Bigr)^{-}\Bigr]<+\infty,$$
we deduce that $U\left(t,x{{\cal E}}_t(\theta\cdot S)\right)$ is a supermartingale for any admissible strategy $\theta$ and any $x>0$. This ends the proof of the theorem.\qed

\begin{remark} 1) From the proof of the theorem, one can easily see that (\ref{Ddecomposition}) is fulfilled when $U_p$ is as forward utility.\\
  2) Theorems \ref{power-characterization-ass} and \ref{pconstant} state that the set of forward utility having the form of (\ref{harapower}) is parameterized by the pair $\left(p(0), Z^D\right)$, where $Z^D$ is the positive local martingale in the multiplicative Doob-Meyer decomposition of $D$ that exists when  $U_p$ is a forward utility. If we suppose that $p(0)$ is given, then Theorem \ref{power-characterization-ass} claims that all power-type forward utilities are obtained by combining the local change of probability ---that corresponds to the local change of belief economically--- and Theorem \ref{case-pre+var-Hellinger}. Again, the optimal portfolio rate is explicitly described by (\ref{root-theta-general0}) once the parameters $\left(p(0), Z^D\right)$ are chosen.
\end{remark}

\subsection{Proof of Theorem \ref{case-pre+var-Hellinger}}\label{sectyionproofofTheorem4}

The proof of Theorem \ref{case-pre+var-Hellinger} requires three intermediary technical lemmas. The main tools used in the proof of these lemmas are two types of integrations ---involved with the random measure $\mu$--- that we precise in the following.

\begin{definition}\label{integrationswithmu} Let $K=(K(t,\omega,x),\ \ x\in\mathbb R^d,\ \omega\in\Omega,\ t\in [0,T])$ be a $\widetilde{\cal P}$-measurable functional.\\
(1) If $K$ is nonnegative, we denote by $K\star\mu$ and $K\star\nu$) the following nondecreasing processes
\begin{equation}\label{Wstarmu}\begin{array}{llll}
\left(K\star\mu\right)_t:=\displaystyle\int_0^t \int_{\mathbb R^d}K(u,x)\mu(du,dx)=\sum_{0<u\leq t}K(u,\Delta S_u)I_{\{ \Delta S_u\not=0\}}\\
\\ \mbox{and}
 \ \ \ \ \left(K\star\nu\right)_t:=\displaystyle\int_0^t \int_{\mathbb R^d}K(u,x)\nu(du,dx).\end{array}\end{equation}
(2) Let $Q$ be a probability, and $\nu^Q$ is the compensator of $\mu$ under $Q$. We say that $K$ is $(\mu-\nu^Q)$-integrable, and denote $K\in {\cal G}^1_{loc}(\mu,Q)$, if
\begin{equation}\label{munuintegrability}
 \left(\sum_{0<t\leq\cdot}\Bigl(K(t,\Delta S_t)I_{\{\Delta S_t\not=0\}}-\int K_t(x)\nu^Q(\{t\},dx)\Bigl)^2\right)^{1/2}\in {\cal A}^+_{loc}(Q).\end{equation}
 In this case, the resulting integral ---denoted by $K\star(\mu-\nu^Q)$--- is a $Q$-local martingale. When $Q=P$, we simply write $W\in {\cal G}^1_{loc}(\mu)$.
\end{definition}

For more details and properties about these two integrations using random measures as well as the obtained integrals, we refer the reader to Jacod (1979), Jacod and Shiryaev (2003), or He, Wang and Yan and (1992).

\begin{lemma}\label{lemmafrorstep(b)}
Suppose that $S$ is locally bounded, ${\cal Z}_{loc}^e(S)\not=\emptyset$, Assumptions \ref{crucialassumption} with $M\equiv 1$ holds, $D$ is predictable with finite variation, and $U_p$ given by (\ref{harapower}) is a forward utility. Then, the process $D$ satisfies
\begin{equation}\label{Drepresentation}
D=D_0\exp(a^D)=D_0{{\cal E}}(X^D) ,\hskip 1cm  X^D:=a^D+\sum (e^{\Delta a^D}-1-\Delta a^D),\end{equation}
and the following assertions hold.\\
(i) For any $\alpha\in (0,1)$, the processes
\begin{equation}\label{processestetahat}
{{(1+\widehat\theta^{tr} z)^p-1-p \widehat\theta^{tr} z}\over{p(p-1)}}I_{\{ \vert \widehat\theta^{tr} z\vert\leq \alpha\}}\star\mu,\ \
 \ \mbox{and}\ \ \ (1+\widehat\theta^{tr} z)^pI_{\{ \vert \widehat\theta^{tr} z\vert> \alpha\}}\star\mu,\end{equation}
are non decreasing and locally integrable.\\
(ii) $P\otimes A$-almost all $(\omega,t)\in \Omega\times [0,+\infty[$, $\widehat\theta\in int\left(\mathcal{D}^+\right)$.\\
(iii) The optimal portfolio rate, $\widehat\theta$, is a root of (\ref{eq-root-thetatilde}). That is
\begin{equation}\label{martingalequationforthethat}
0={1\over{p-1}}b+c\widehat\theta+\int {{(1+\widehat\theta^{tr} z)^{p-1}-1}\over{p-1}}z F(dz),\ \ \ \ \ \ P\otimes A-a.e.\end{equation}
(iv) The optimal portfolio rate, $\widehat\theta$, satisfies
\begin{equation}\label{hellinger=D}
e^{-\Delta a^{D}}\cdot X^{D}=\displaystyle\frac{p(p-1)}{2}\widehat{\theta}^{tr}c\widehat{\theta}\cdot A-\left[\int\left[(1+\widehat{\theta}^{tr}z)^{p-1}-1+(1-p){\widehat\theta}^{tr}x(1+\widehat{\theta}^{tr}z)^{p-1}\right]F(dz)\right]\cdot A.
\end{equation}
Here $X^D$ is given by (\ref{Drepresentation}).\\
(v) If we denote $u(t,\omega,x):=\left(1+x^{tr}\widehat\theta_t(\omega)\right)^{p-1}$, then
\begin{equation}\label{vhat}
1-a+\widehat u_t:=1+\int (u(t,x)-1)\nu(\{t\},dx)=\exp(-\Delta a^D),\ \ \ \ \ a_t:=\nu(\{t\},{{\it \hbox{I\kern-.2em\hbox{R}}}}^d),\end{equation}
and as a consequence, the positive predictable process, $\left(1-a+\widehat u\right)^{-1}$, is locally bounded.
\end{lemma}

{\bf Proof.} Since $pD(t)>0$ for all $(t,\omega)\in [0,T]\times\Omega$ ---$U(t,x)$ is a random field utility---, it is obvious to see that $D(t)/D_0$ is a positive and predictable process with finite variation. Therefore, the decomposition in (\ref{Drepresentation}) holds. The proof of the assertions (i)---(v) of the lemma will be carried in three steps (parts a),  b), and c) below). Part a) proves assertion (i), while part b) proves assertions (ii), (iii) and (iv). The part c) proves assertion (v).\\
{\bf a)} Remark that
\begin{equation}\label{U-theta}\begin{array}{lll}
U(t,x{\cal E}_{t}(\theta\cdot S))=D_{0}x^{p}\exp\left(a^{D}_{t}\right){\cal E}_{t}\left(\theta\cdot S\right)^{p}
=D_{0}x^{p}{\cal E}_{t}\left(X^{D}\right){\cal E}_{t}\left({X}^{\theta}\right)\\
\\
\hskip 2.5cm =D_{0}x^{p}{\cal E}_{t}\left((1+\Delta X^D)\cdot {X}^{\theta}+X^D\right),\end{array}
\end{equation}
for any admissible strategy $\theta$, where $X^D$ is defined in (\ref{Drepresentation}) and $X^{\theta}$ is given by
\begin{equation}\label{Xtheta}
{X}^{\theta}=p\theta\cdot S+\frac{p(p-1)}{2}\theta^{tr}c\theta\cdot A+\left((1+\theta^{tr}z)^{p}-1-p\theta^{tr}z\right)\star\mu.
\end{equation}
Thus, $U(t,x{\cal E}_{t}(\theta\cdot S))$ is a local supermartingale (local martingale for $\theta=\widehat\theta$) if and only if
\begin{equation}\label{supermartinagelforXtheta}
{{1}\over{p(1-p)}}\left(e^{\Delta a^{D}}\cdot {X}^{\theta}+X^{D}\right)\ \ \mbox{is a local supermartingale (local martingale for }\theta=\widehat\theta).\end{equation}
Due to Ito's formula, we easily deduce that (\ref{supermartinagelforXtheta}) holds if and only if
\begin{equation}\label{conditionA}
\big|(1+\theta^{tr}z)^{p}-1-p\theta^{tr}z I_{\{\vert \theta^{tr}z\vert\leq \alpha\}} \big|\star\mu\in{\cal A}^+_{loc},\end{equation}
and $P\otimes A$-almost all $(\omega, t)$, we have
\begin{equation}\label{conditionB}
{{\exp(-\Delta a^{D})}\over{p(1-p)}}\cdot X^{D}=\Phi_p\left(\widehat{\theta}\right)\cdot A,\ \ \  \mbox{and}\quad\quad\mathop{\min}_{\theta\in\mathbb R^d}\left[\Phi_p(\theta)\right]=\Phi_p(\widehat\theta),
\end{equation}
where $\Phi_p$ is given by (\ref{Phip}), that we recall below for the convenience of the reader
\begin{equation}\label{Phipfunction}
\Phi_p(\theta):={{b^{tr}\theta}\over{p-1}}+{1\over{2}}\theta^{tr}c\theta+\int{{(1+\theta^{tr} x)^p-1-p\theta^{tr} x}\over{p(p-1)}}F(dx).\end{equation}

\noindent Due to $I_{\{\vert \theta^{tr}z\vert> \alpha\}} \star\mu=\sum I_{\{\vert \theta^{tr}\Delta S\vert> \alpha\}}\in {\cal A}^+_{loc}$, and
$$\begin{array}{lll}
 \big|(1+\theta^{tr}z)^{p}-1-p\theta^{tr}z I_{\{\vert \theta^{tr}z\vert\leq \alpha\}} \big|\star\mu=\big|(1+\theta^{tr}z)^{p}-1-p\theta^{tr}z  \big| I_{\{\vert \theta^{tr}z\vert\leq \alpha\}}\star\mu+\\
 \\
 \hskip 6.5cm +\big|(1+\theta^{tr}z)^{p}-1\big| I_{\{\vert \theta^{tr}z\vert> \alpha\}} \star\mu,\end{array}$$
 we deduce that assertion (i) of the lemma follows from (\ref{conditionA}).\\
{\bf b)} By combining the second equality in (\ref{conditionB}) with Lemma \ref{lem-interior-integrability}, both assertion (ii) and (iii) of the lemma follows immediately, while assertion (iv) follows from inserting (\ref{martingalequationforthethat}) into the first equation of (\ref{conditionB}).\\
{\bf c)} By multiplying (\ref{martingalequationforthethat}) with $\Delta A$, and using $\Delta A b=\int x F(dx)\Delta A$, $\Delta A c=0$, and $F_t(dx)\Delta A_t=\nu(\{t\},dx)$, we get
 \begin{equation}\label{martingalejump}
  \displaystyle
\int (1+\widehat\theta^{tr}z)^{p-1}z \nu(\{t\},dz)=0,\end{equation}
  on the one hand. On the other hand, by taking jumps in (\ref{hellinger=D}), inserting (\ref{martingalejump}) in the resulting equation afterwards, and using again $\Delta Ac=0$, we obtain
$$
 1-\exp(-\Delta a^{D})= \exp(-\Delta a^{D})\Delta X^{D}=-\int(1+\widehat{\theta}^{tr}z)^{p-1}\nu(\{t\},dz)+a=-\widehat u+a.$$
  Thus, assertion (v) follows from the above equality and the local boundeness of $e^{-\Delta a^D}$.\\
This achieves the proof of the lemma. \qed


\noindent The following lemma will show how the minimal Hellinger martingale density of order $q$ is built-up and is related to the optimal portfolio rate, $\widehat\theta$, when $U_p$ is a forward utility.

\begin{lemma}\label{lemmafrorstep(c)} Suppose that $S$ is locally bounded, ${\cal Z}_{loc}^e(S)\not=\emptyset$, Assumptions \ref{crucialassumption} with $M\equiv 1$ holds, $D$ is predictable with finite variation, and $U_p$ given by (\ref{harapower}) is a forward utility. Then, the following properties hold:\\
(i) The $\widetilde{\cal P}$-measurable functional
\begin{equation}\label{W}
W_{t}(z):=\frac{\left(1+\widehat{\theta}_{t}^{tr}z\right)^{1/(q-1)}-1}
{1-a_{t}+\int\left(1+\widehat{\theta}_{t}^{tr}y\right)^{1/(q-1)}\nu(\{t\},dy)}=:{{u(t,z)-1}\over{1-1+\widehat u}},\end{equation}
is $(\mu-\nu)$-integrable, (i.e. $W\in {\cal G}^1_{loc}(\mu)$ see (\ref{munuintegrability}) for the definition of this set).\\
(ii) The process, $\widetilde Z$, defined by
\begin{equation}\label{Ztilde}
\widetilde{Z}:={\cal E}\left(\widetilde N\right),\quad \widetilde N:={1\over{q-1}}\widehat{\theta}\cdot S^{c}+W\star(\mu-\nu),\end{equation}
is a martingale density for $S$.\\
(iii) The following
\begin{equation}\label{HellingerXD}
(q-1)X^D+\sum\Bigl[\left(1+\Delta X^D\right)^{q-1}-1-(q-1)\Delta X^D\Bigr]=q(q-1)h^{(q)}(\widetilde Z, P)\end{equation}
holds, where $X^D$ is defined in (\ref{Drepresentation}).
\end{lemma}

{\bf Proof.} Thanks to Lemma \ref{lemmafrorstep(b)}--(v), $\left(\widetilde\gamma_t\right)^{-1}:=\left(1-a+\widehat u\right)^{-1}=e^{\Delta a^D}$ is locally bounded, and
$$
\sum_{0<t\leq\cdot} (\widehat W_t)^2:=\sum_{0<t\leq\cdot}\left(\int W_t(x)\nu(\{t\},dx)\right)^2=\sum_{0<t\leq\cdot} \left({{\widehat u_t}\over{1-a+\widehat u_t}}\right)^2\preceq e^{3\vert\Delta a^D\vert}\cdot \vert a^D\vert_{var}$$ is a nondecreasing and locally bounded process. Therefore, $W\in {\cal G}^1_{loc}(\mu)$ if and only if $$
\left[\sum \left(W_t(\Delta S_t)\right)^2 I_{\{ \Delta S_t\not=0\}}\right]^{1/2}=\left[(\widetilde\gamma)^{-2}\Bigl((1+\widehat{\theta}^{tr}x)^{p-1}-1\Bigr)^2\star\mu\right]^{1/2}\in {\cal A}^+_{loc}(P).$$ Again, the boundedness of $(\widetilde\gamma)^{-2}=(1-a+\widehat u)^{-2}=e^{2\Delta a^D}$ (see assertion (v) of Lemma \ref{lemmafrorstep(b)}) implies that the above claim is equivalent to
\begin{equation}\label{localintegrable}
\left[\Bigl((1+\widehat{\theta}^{tr}x)^{p-1}-1\Bigr)^2\star\mu\right]^{1/2}\in {\cal A}^+_{loc}(P).\end{equation}
If we put $\Gamma:=\{z\in \mathbb R^d\Big| \ \vert \widehat\theta^{tr} z\vert\leq\alpha\}$, then ---due to $(\sum (\Delta X)^{2})^{1/2}\leq \sum |\Delta X|$--- it is easy to check that (\ref{localintegrable}) is implied by the local integrability of
 \begin{equation}\label{V1V2}
 V_1:=\Bigl((1+\widehat{\theta}^{tr}z)^{p-1}-1\Bigr)^2 I_{\Gamma}\star\mu,\ \mbox{and}\  \ V_2:=\vert
(1+\widehat{\theta}^{tr}z)^{p-1}-1\vert I_{\Gamma^c}\star\mu.\end{equation}
The local integrability of $V_1$ follows directly from $ I_{\{\vert\widehat\theta^{tr}\Delta S\vert\leq
\alpha\}}\cdot[\widehat\theta\cdot S, \widehat\theta\cdot
S]\in{\cal A}^+_{loc}$ (since $\widehat{\theta}$ is $S$-integrable and hence $\widehat\theta\cdot S$ is a c\`adl\`ag semimartingale), and
\begin{eqnarray*}
\frac{(q-1)^{2}}{(1-\alpha)^{2(p-2)}}V_1 &\preceq& \sum (\widehat\theta^{tr}\Delta S)^2
I_{\{\vert\widehat\theta^{tr}\Delta S\vert\leq \alpha\}}
\preceq
I_{\{\vert\widehat\theta^{tr}\Delta S\vert\leq
\alpha\}}\cdot[\widehat\theta\cdot S, \widehat\theta\cdot
S].\end{eqnarray*}

\noindent To prove the local integrability of $V_2$, it is enough to prove
\begin{equation}\label{integro1}
\vert \int (1+\widehat\theta^T z)^{p-1}\widehat\theta^{tr} z I_{\{\vert \widehat\theta^{tr} z\vert >\alpha\}} F(dz)\vert\cdot A\in {\cal A}^+_{loc}.\end{equation}
Indeed, by combining (\ref{integro1}) with $(1+\widehat\theta^{tr} z)^{p}I_{\{\vert \widehat\theta^{tr} z\vert >\alpha\}}\star\mu\in {\cal A}^+_{loc}$ (see Lemma \ref{lemmafrorstep(b)}--(i)), and
$$
(1+\widehat\theta^{tr} z)^{p-1}I_{\Gamma^c}\star\nu=
-\int_{\Gamma^c} (1+\widehat\theta^{tr} z)^{p-1}\widehat\theta^{tr} z F(dz)\cdot A+(1+\widehat\theta^{tr}z)^{p}I_{\Gamma^c}\star\nu,$$
we deduce that $(1+\widehat\theta^{tr} z)^{p-1}I_{\{\vert \widehat\theta^{tr}z\vert >\alpha\}}\star\mu$ is locally integrable (since it is nondecreasing and its compensator is locally integrable). Finally, due to
$$I_{\{\vert \widehat\theta^{tr}z\vert >\alpha\}}\star\mu=\sum I_{\{\vert \widehat\theta^{tr}\Delta S\vert >\alpha\}}\in {\cal A}^+_{loc},$$ which follows from the fact that $\widehat\theta\cdot S$ is a c\`adl\`ag semimartingale, we conclude that $V_2$ is locally integrable. In the remaining part of this proof, we will prove (\ref{integro1}). Thanks to Proposition \ref{theta-Doobdecomposition}, we have
\begin{equation}\label{equa1000}
\widehat\theta^{tr} c\widehat\theta\cdot A\in{\cal A}^+_{loc},\ \ \mbox{and}\ \ \vert\widehat\xi\vert\cdot A:=\vert\widehat\theta^{tr} b-\int \widehat\theta^{tr} z I_{\{\vert \widehat\theta^{tr} z\vert>\alpha\}} F(dz)\vert\cdot A\in{\cal A}^+_{loc}.\end{equation}
Since $\widehat\theta$ satisfies (\ref{martingalequationforthethat}), then we get (recall that $\Gamma:=\{x\in\mathbb R^d:\ \vert\widehat\theta^{tr} x\vert\leq \alpha\}$ and $q-1=(p-1)^{-1}$)
\begin{equation}\label{equa1001}
-(q-1)\widehat\xi-\widehat\theta^{tr} c\widehat\theta=
{1\over{p-1}}\int_{\Gamma^c} (1+\widehat\theta^{tr} z)^{p-1}\widehat\theta^{tr} z F(dz)+ \int_{\Gamma} {{(1+\widehat\theta^{tr} z)^{p-1}-1}\over{p-1}}\widehat\theta^{tr} z F(dz).\end{equation}
Then, by combining
$$\begin{array}{llll}
0\preceq {{(1+\widehat\theta^{tr} z)^{p-1}-1}\over{p-1}}(\widehat\theta^{tr} z) I_{\{\vert \widehat\theta^{tr} z\vert \leq\alpha\}}\star\mu\preceq (1-\alpha)^{p-2}(\widehat\theta^{tr} z)^2 I_{\{\vert \widehat\theta^{tr} z\vert \leq\alpha\}}\star\mu\\
\\
\hskip 5cm \preceq (1-\alpha)^{p-2}I_{\{\vert \widehat\theta^{tr} \Delta S\vert \leq\alpha\}}\cdot [\widehat\theta\cdot S,\widehat\theta\cdot S]\in {\cal A}^+_{loc},\end{array}$$
(\ref{equa1000}) and (\ref{equa1001}), we conclude that (\ref{integro1}) holds. This ends the proof of assertion (i).\\
Thus, $\widetilde Z$ is well defined and is a positive local martingale. Then, a direct application of Ito's formula for $\widetilde ZS$ leads to conclude that $\widetilde ZS$ is a local martingale if and only if
$$
b\cdot A+(p-1)c\widehat\theta\cdot A+\int x\left((1+\widehat\theta^{tr} x)^{p-1}(1-a+\widehat u)^{-1}-1\right)F(dx)\cdot A\equiv 0.$$
Then, it is easy to check ---by distinguishing the two cases whether $\Delta A=0$ or $\Delta A\not=0$--- that the above equation is equivalent to (\ref{martingalequationforthethat}). Indeed, it is clear that we have $b\Delta A=\int xF(dx)\Delta A$ and $c\Delta A=0$, and on $\{\Delta A=0\}$ we have $1-a+\widehat u=1$. This ends the proof of assertion (ii).\\

\noindent In the remaining part of this proof, we focus on proving the last assertion (i.e. assertion (iii)). By combining (\ref{vhat}) and (\ref{HZtilde}) when $Z\equiv 1$, we deduce that
 $$
 (1+\Delta X^D)^{q-1}-1=q(q-1)\Delta h^{(q)}(\widetilde Z,P).$$
 This equation is exactly (\ref{HellingerXD}) on $\{ \Delta A\not=0\}$, while on $\{\Delta A=0\}$ (\ref{HellingerXD}) follows from combining (\ref{hellinger=D}) and (\ref{HellingerAc}) when $Z\equiv 1$ (recall here $\widetilde\beta$ and $\widehat\theta$ coincide). This proves (\ref{HellingerXD}), and the proof of the lemma is completed. \qed


\begin{remarks}\label{remarkforp/zero}
1) In the proof of Lemma \ref{lemmafrorstep(c)}, it is easy to notice that the proof of (\ref{localintegrable}) follows exactly from $\widehat\theta\in L(S)$ and the equation (\ref{martingalequationforthethat}) that $\widehat\theta$ satisfies. Therefore, the proof is also valid for the case of $p=0$. Thus, if $\widehat\theta\in L(S)$ and is a root of
\begin{equation}\label{martinbgaleforthetahatzero}
b-c\lambda+\int \left[\left(1+\lambda^{tr} x\right)^{-1}-1\right]x F(dx)=0,\end{equation}
 then the process $\Bigl[\left(1+\widehat\theta^{tr}x\right)^{-1}-1\Bigr]\in {\cal G}^1_{loc}(\mu)$ .\\
2) Also the proof of Lemma \ref{lemmafrorstep(e)} (see below) is based on (\ref{martingalequationforthethat}) that $\widehat\theta$ fulfills and the form of $\widetilde Z$ given by (\ref{Ztilde}) only. These two ingredients do not assume any condition on $p\in (-\infty,1)$. As a result, Lemma \ref{lemmafrorstep(e)} is still valid for the case of $p=0$, or equivalently $\widetilde Z$ is the minimal Hellinger martingale density of order zero as long as $\widehat\theta\in L(S)$ solves (\ref{martingalequationforthethat}) and $\widetilde Z$ is given by (\ref{Ztilde}) when $p=0$.
\end{remarks}


\begin{lemma}\label{lemmafrorstep(e)}
The process $\widetilde Z$ defined in Lemma \ref{lemmafrorstep(c)} is the minimal Hellinger martingale density of order $q$. That is, $\widetilde Z$ is a martingale density (belongs to ${\cal Z}^e_{loc}(S,P)$) satisfying
\begin{equation}\label{goaloflemma}
h^{(q)}(\widetilde Z,P)\preceq h^{(q)}(Z,P),\end{equation}
for any $ Z\in{\cal Z}^e_{loc}(S,P)$.
\end{lemma}

{\bf Proof.} Thanks to Lemma  \ref{lemmafrorstep(c)}--(ii), the proof of the lemma will follow from proving the optimality of $\widetilde
Z$. In virtue of Proposition
3.2 in Choulli and Stricker (2006) (see also Proposition 4.2 in
Choulli and Stricker (2005) for the case of quasi-left continuity), it
is enough to prove that (\ref{goaloflemma}) holds for any positive martingale density
$Z={{\cal E}}(N)$ of the form
$$
N=\beta\cdot S^c+Y\star(\mu-\nu), \ \ Y_t(x)=k_t(x)+{{\widehat
k_t}\over{1-a_t}}I_{\{a_t<1\}}, \ \ \widehat k_t:=\int
k_t(x)\nu(\{t\},d x),$$ where $\beta\in L(S)$ and $\left(\sum
k_t(\Delta S_t)^2I_{\{\Delta S_t\not=0\}}\right)^{1/2}\in{\cal
A}^+_{loc}$. Due to the convexity of $z^T cz$ and
$\phi(z):=\frac{(1+z)^{q}-qz-1}{q(q-1)}$, on the set $\{\Delta A=0\}$ we derive
\begin{equation}\label{equation111}\begin{array}{llll}
\displaystyle{{d h^{(q)}(Z,P)}\over{d A}}-{{d h^{(q)}(\widetilde Z,P)}\over{d A}}={1\over{2}}(\beta^{tr} c\beta-\widetilde\theta^{tr} c\widetilde\theta)+\int\Bigl[\phi(k(x))-\phi(\widetilde k(x))\Bigr]F(d x)\\
\\
\hskip 4cm \geq \widetilde\theta^T
c(\beta-\widetilde\theta)+\displaystyle\int\widetilde\theta^{tr}x
\Bigl(k(x)-\widetilde k(x)\Bigr)F(d
x)=0.\end{array}\end{equation}
Here $\widetilde\theta=(p-1)\widehat\theta$, $\widetilde k(x):=(1+\widehat\theta^{tr}x)^{p-1}-1$ and $\phi'(k((x))=(p-1)\widehat\theta^{tr} x=\widetilde\theta^{tr} x$. The last equality in (\ref{equation111}) is obtained from the fact that both $\widetilde Z$ and $Z$ belong to ${\cal Z}_{loc}^{e}(S)$, which is equivalent to
 \begin{equation}\label{martingaleRecall}
 b+c\beta+\int xk(x)F(dx)=0,\ \ \mbox{and}\ \ b+(p-1)c\widehat\theta+\int x\widetilde k(x)F(dx)=0.\end{equation}
Now, we compare the two jump processes $\Delta h_t^E(Z,P)$ and $\Delta h_t^E(\widetilde Z,P)$, by using again the convexity of $\phi(z)$, as follows
\begin{eqnarray}
&&\Delta h_t^E(Z,P)-\Delta h_t^E(\widetilde Z,P)=(1-a_t)\left[\phi\Bigl(-{{\widehat k_t}\over{1-a_t}}\Bigr)-\phi\Bigl({\widetilde\gamma_t}^{-1}-1\Bigl)\right]\notag\\
&+&\displaystyle\int\Bigl[\phi(k_t(x))-\phi\Bigl((1+\widehat\theta_t^{tr} x)^{p-1}{\widetilde\gamma}_{t}^{-1}-1\Bigr)\Bigr]\nu_{t}(d x)\nonumber\\
&\geq& (1-a_t)(1-{{\widehat k_t}\over{1-a_t}}-{1\over{\widetilde\gamma_t}})\frac{\widetilde{\gamma}_{t}^{1-q}-1}{q-1}\nonumber\\
&+&\displaystyle\int\Bigl[k_t(x)+1-{\widetilde\gamma_t}^{-1}(1+\widehat\theta_t^{tr}
x)^{p-1}\Bigr]\frac{(\widehat\theta_t^{tr}
x+1)\widetilde{\gamma}_{t}^{1-q}-1}{q-1}\nu_{t}(d x)\nonumber\\
&=& \frac{\widetilde{\gamma}_{t}^{1-q}}{q-1}\displaystyle\int\Bigl[
(k_t(x)+1)-(\widetilde\gamma_t)^{-1}(1+\widehat{\theta}^{tr}x)^{p-1}\Bigr]\widehat\theta_t^{tr} x\nu_{t}( d
x)=0.\label{equation112}\end{eqnarray}
The last equality in (\ref{equation112})
follows from (\ref{martingaleRecall}) (by multiplying both equations with $\Delta A$ and using $b\Delta A=\int x\nu(\{t\},dx)$, $ \Delta A c=0$ and $F_t(dx)\Delta A_t=\nu(\{t\},dx)$ ) that leads to
 $$
 \int x(k_t(x)+1)\nu(\{t\},dx)=0,\ \ \mbox{and}\ \ 0={\widetilde\gamma}^{-1}_t\int x(1+\widehat\theta^{tr}_t x)^{p-1}\nu(\{t\},dx).$$ Thus, by combining
(\ref{equation111}) and (\ref{equation112}), we deduce that
$\widetilde Z$ is the minimal Hellinger martingale density of order $q$ for $S$. This achieves the proof of the lemma. \qed

Now, we are ready to provide the proof of Theorem \ref{case-pre+var-Hellinger}.\\

{\bf Proof of Theorem \ref{case-pre+var-Hellinger}:} We start proving $(1)\Longrightarrow(2)$. Thus, suppose that assertion (1) holds. Therefore, Lemmas \ref{lemmafrorstep(b)}, \ref{lemmafrorstep(c)}, and \ref{lemmafrorstep(e)} are valid, and the minimal Hellinger martingale density, $\widetilde Z$ exists (it is given by Lemma \ref{lemmafrorstep(c)}). Furthermore, an application of Ito's formula to ${\cal E}(X^D)^{q-1}$ combined with (\ref{Drepresentation}) and (\ref{HellingerXD}), will easily lead to  (\ref{pre+var-characterD(t)}). This proves assertions (2.b) and (2.c) of the theorem. To conclude that assertion (2) is satisfied, we need to prove assertion (2.a). This follows from the forward property of $U_p$ with optimal portfolio rate $\widehat\theta$ and
\begin{eqnarray}
U_p\left(\cdot,x{\cal E}\left(\widehat{\theta}\cdot S\right)\right)
&=&D_{0}x^{p}{\cal E}\left(q(q-1)h^{(q)}\left(\widetilde{Z},P\right)\right)^{p-1}{\cal E}\left(\widehat{\theta}\cdot S\right)^{p}\label{cont-thetatilde-Utrue0}\\
&=&D_{0}x^{p}{\cal E}\left(\widehat{\theta}\cdot S\right){\cal E}\left(\widetilde\gamma^{1-q}\widehat{\theta}
\cdot S+q(q-1)h^{(q)}\left(\widetilde{Z},P\right)\right)^{p-1}\label{cont-thetatilde-Utrue1}\\
&=&D_{0}x^{p}{\cal E}\left(\widehat{\theta}\cdot S\right)\widetilde{Z}=D_{0}x^{p}\widehat{Z}.\label{cont-thetatilde-Utrue2}
\end{eqnarray}
It is clear that (\ref{cont-thetatilde-Utrue0}) follows from (\ref{pre+var-characterD(t)}) and $p-1={1\over{q-1}}$, while (\ref{cont-thetatilde-Utrue1})  and (\ref{cont-thetatilde-Utrue2}) follows from (\ref{ZtildeZ}) ---for the case of $Z\equiv 1$--- whenever the MHM density of order $q$, $\widetilde{Z}$, exists. This proves assertion (2).\\

In the remaining part of this proof, we focus on proving $(2)\Longrightarrow(1)$. Thus, we suppose that assertion (2) is fulfilled. Then, it is obvious that (\ref{cont-thetatilde-Utrue0}), (\ref{cont-thetatilde-Utrue1}), and (\ref{cont-thetatilde-Utrue2}) always hold as long as the MHM density of order $q$ exists and assertion (2.b) is valid. As a consequence, a combination of these equalities with assertion (2.a) imply that $U_p\left(\cdot,x{\cal E}\left(\widehat{\theta}\cdot S\right)\right)$ is a martingale. Furthermore, for any admissible $\theta$, we have
$$\begin{array}{llll}
\displaystyle (D_{0}x^{p})^{-1}U_p\left(\cdot,x{\cal E}\left({\theta}\cdot S\right)\right)
=\displaystyle {\cal E}\left({\theta}\cdot S\right)^{p}{\cal E}\left(\widehat{\theta}
\cdot S\right)^{1-p}{\cal E}\left(q(q-1)h^{(q)}\left(\widetilde{Z},P\right)\right)^{p-1}{\cal E}\left(\widehat{\theta}
\cdot S\right)^{p-1}\notag\\
\hskip 4cm =\displaystyle{\cal E}\left(\widehat{\theta}\cdot S\right)\widetilde{Z}\left({\cal E}\left({\theta}\cdot S\right)/{\cal E}\left(\widehat{\theta}
\cdot S\right)\right)^{p}\notag=\displaystyle \widehat{Z}\left({\cal E}\left({\theta}\cdot S\right)/{\cal E}\left(\widehat{\theta}
\cdot S\right)\right)^{p}.
\end{array}$$
Due to this equality, the equivalence between $\theta\in \Theta(x,U_p)$ and (\ref{p<0-integrability}), and  Proposition \ref{theta-supermatingale}, we conclude that $U_p\left(\cdot,x{\cal E}\left({\theta}\cdot S\right)\right)$ is a supermartingale for any admissible portfolio rate $\theta$. Hence, $U_p$ is a forward utility and assertion (1) holds. This ends the proof of the theorem. \qed



\section{Parametrization of the Log-Type Forward Utilities}\label{subsection:logcase}

This section focuses on describing the forward utilities having the form of 
\begin{equation}\label{haralog}
U_0(t,x):={\widehat D}(t)\log\left(x\right)+\overline{D}(t).\end{equation}
 We will prove that this class of forward utilities is completely parameterized by two local martingales intimately related to ${\widehat D}$ and ${\overline D}$. For these random field utilities, we consider the set of admissible portfolios, ${\cal A}_{adm}(x,\log)$, that is slightly different than the one defined in (\ref{admissibleset}).
$$
{\cal A}_{adm}(x,\log):=\Bigl\{\pi\in L(S)\ \big|\  \ x+\pi\cdot S> 0,\ \&\
\left(\left[U_0(\tau, x+(\pi\cdot S)_{\tau})\right]^-\right)_{ \tau\in {\cal T}_T}\ \mbox{is uniformly integrable}.\Bigr\}$$
Then, the set of admissible portfolio rates ---denoted by $\Theta(x,\log)$--- is given by
\begin{equation}\label{admissiblelog}
\Theta(x,\log):=\Bigl\{\theta\in L(S)\ \big|\ \ {\cal E}(\theta\cdot S)>0\ \ \&\ \ \ \theta{\cal E}_{-}(\theta\cdot S)\in {\cal A}_{adm}(x,\log)\Bigr\}.
\end{equation}

Below, we elaborate our main result of this section.
\begin{theorem}\label{forwardlog} Suppose that $U_0$ defined in (\ref{haralog}) is a random field utility such that
\begin{equation}\label{log-null-strategy}
\sup_{\tau\in\mathcal{T}_{T}}E\left(|{\widehat D}(\tau)|\right)<+\infty\ \ \mbox{and} \ \ \sup_{\tau\in\mathcal{T}_{T}}E\left(|\overline{D}(\tau)|\right)<+\infty.
\end{equation}
Suppose that  $S$ is locally bounded, ${\cal Z}_{loc}^e(S)\not=\emptyset$, and Assumption \ref{crucialassumption} with $M_t\equiv E\left({\widehat D}(T)\ |\ {\cal F}_t\right)/E({\widehat D}(T))$ holds.
Then the assertions (1) and (2) below are equivalent.\\
(1) The functional $ U_0$ is a forward utility with the optimal portfolio rate $\widehat{\theta}$.\\
(2) The following properties hold:\\
(2.a) The process $\widehat D$ is a positive martingale.\\
(2.b) The MHM density of order zero with respect to $Q:={\widehat D}(T)\left({\widehat D}(0)\right)^{-1}\cdot P$  exists (that we denote by $\widetilde Z^{Q}$), and there exists
 a $Q$-local martingale $L^Q$ such that
\begin{equation}\label{DforLogUtility}
\overline{D}(t)={\widehat D}(t)\left(\overline{D}(0)\left({\widehat D}(0)\right)^{-1}+L_t^Q-h^{(0)}_t(\widetilde Z^{Q},Q)\right),\ \ \ \ \ \ \ 0\leq t\leq T.\end{equation}
(2.c) $P\otimes A-$almost all $(\omega,t)$, the optimal portfolio rate $\widehat\theta$ belongs to $\mbox{int}\left({\cal D}\right)$, and is a root for
\begin{equation}\label{martingaleofzero}
b^{Q}-c\lambda+\int \left[\left(1+\lambda^{tr} x\right)^{-1}-1\right]x F^{Q}(dx)=0,\end{equation}
where $b^{Q}$ and $F^{Q}(dx)$ are the predictable characteristics of $S$ under $Q$.\\
(2.d) The process $\widehat N_t:=\overline{D}(t)-{\widehat D}(t)\log\left(\widetilde Z^Q_t\right)$ is a martingale.
\end{theorem}

{\bf Proof.}  We start proving the difficult part, which is $(1)\Longrightarrow(2)$. Suppose that assertion (1) holds. Thanks to Proposition \ref{optimalportfolio}, the optimal portfolio rate $\widehat\theta_x$ does not depend on the initial capital $x\in (0,+\infty)$. Thus, a combination of this fact with assertion (1) lead to put $\widehat\theta:=\widehat\theta_1$, and
\begin{equation}\label{martingaleequationforlog}
U_0\left(t,x{\cal E}_t(\widehat\theta\cdot S)\right)=\log(x){\widehat D}(t)+{\widehat D}(t)\log\left({\cal E}_t(\widehat\theta\cdot S)\right)+\overline{D}(t)\ \ \mbox{is a martingale}, \end{equation}
for any $x\in (0,+\infty)$, and for any $\theta\in\Theta(x,\log)$,
\begin{equation}\label{supermartingaleforlog}
U_0\left(t,x{\cal E}_t(\theta\cdot S)\right)=\log(x){\widehat D}(t)+{\widehat D}(t)\log\left({\cal E}_t(\theta\cdot S)\right)+\overline{D}(t)\ \ \mbox{is a supermartingale.}
\end{equation}
 Thus, ${\widehat D}(t)$ is a positive martingale, and $\overline D$ is a c\`adl\`ag supermartingale. This proves assertion (2.a) as well as $\overline{D}(t)/{\widehat D}(t)$ is a supermartingale under $Q:=\left(\widehat D(T)/\widehat D(0)\right)\cdot P$. Thus, there exists a predictable and nondecreasing process, $\overline{A}^{Q}$, such that (\ref{martingaleequationforlog}) and (\ref{supermartingaleforlog}) translate into (after applying Ito's formula and compensating under $Q$)
\begin{equation}\label{optimizationforLog}
-\Phi_0^Q(\theta)\cdot A-\overline{A}^{Q}\preceq -\Phi_0^Q(\widehat\theta)\cdot A-\overline{A}^{Q}=0.\end{equation}
Here the function $\Phi_0^Q$ is given by
(\ref{Phip}) by taking $p=0$ and $R=Q$ and $(b^Q, c, F^Q)$ are the predictable characteristics of $S$ under $Q$ (it is obvious that $c^Q=c$).
 Therefore, (\ref{optimizationforLog}) implies that $\widehat\theta$ minimizes $\Phi_0^Q$ over the set ${\cal D}^+$,  and thus $\widehat\theta$ fulfills the assumptions of Lemma \ref{lem-interior-integrability}. Hence, assertion (2.c) follows immediately from this lemma. Then, thanks to the proofs of Lemma \ref{lemmafrorstep(c)} and Lemma \ref{lemmafrorstep(e)} (See Remarks \ref{remarkforp/zero} for details), we deduce that
$$\widetilde W_0(t,y):=\left[\left(1+\widehat\theta^{tr} y\right)^{-1}-1\right]\in {\cal G}^1_{loc}(\mu, Q),$$
and the minimal Hellinger martingale density  of order zero with respect to $Q$, denoted $\widetilde Z^Q$, exists and is given by
$$
\widetilde Z^Q:={\cal E}(\widetilde N^Q),\ \ \ \ \  \widetilde N^Q:=-\widehat\theta\cdot S^{c,Q}+{\widetilde W}_0\star(\mu-\nu^Q).$$
This proves assertion (2.b). Due to Proposition \ref{lem-represofMHE-changemeasure} (put $q=0$ in (\ref{ZtildeZ})), we obtain
$$
\left(\widetilde Z^Q\right)^{-1}={\cal E}\left(\widehat\theta\cdot S\right).$$
By inserting this equality into (\ref{martingaleequationforlog}) we derive
\begin{equation}\label{martingaleforlkog1}
\begin{array}{lll}
U_0\left(t, x{\cal E}_t(\widehat\theta\cdot S)\right):={\widehat D}(t)\log(x)+{\widehat D}(t)\log\left({\cal E}_t(\widehat\theta\cdot S)\right)+\overline{D}(t)\\
\\
\hskip 2.5cm ={\widehat D}(t)\log(x)+\overline{D}(t)-{\widehat D}(t)\log\left(\widetilde Z^Q_t\right).
\end{array}\end{equation}

\noindent Then, assertion (2.d) follows from the fact that both $U_0\left(t, x{\cal E}_t(\widehat\theta\cdot S)\right)$ and ${\widehat D}$ are martingales. This completes the proof assertion (2).\\

To prove the reverse implication (i.e. $(2)\Longrightarrow (1)$), it is easy to remark that (\ref{martingaleforlkog1}) is valid as long as assertion (2.b) holds. Then, assertions (2.a) and (2.d) implies that $ U^Q_0\left(t, x{\cal E}_t(\widehat\theta\cdot S)\right)$ is a $Q$-martingale for any $x>0$. Then, assertion (1) will follow immediately once we prove that $U^Q_0\left(t, x{\cal E}_t(\theta\cdot S)\right)$ is a $Q$-supermartingale for any $x>0$ and any $\theta\in\Theta(x,\log)$. To this end, we first calculate
$$
\begin{array}{lll}
U^Q_0\left(t, x{\cal E}_t(\theta\cdot S)\right)=\log(x)+\log\left({\cal E}_t(\theta\cdot S)\right)+\overline{D}(t)/{\widehat D}(t)\\
\\
\hskip 2.5cm =\displaystyle \log\left({\cal E}_t(\theta\cdot S)/{\cal E}_t(\widehat\theta\cdot S)\right)+U^Q_0\left(t, x{\cal E}_t(\widehat\theta\cdot S)\right),\\
\\
\hskip 2.5cm=:\log\Bigl(X_0(t)\Bigr)+U^Q_0\left(t, x{\cal E}_t(\widehat\theta\cdot S)\right).\end{array}$$
Then, it is easy to see that $X_0={\cal E}(\theta\cdot S)/{\cal E}(\widehat\theta\cdot S)={\widetilde Z}^Q{\cal E}(\theta\cdot S)$ is a positive $Q$-local martingale (which implies that $\log(X_0(t))$ is a $Q$-local supermartingale) , and ---due to $e^{y^+}\leq e^y+1$---
$$E^Q\left(e^{(\log(X_0(\tau)))^+}\right)\leq E^Q(X_0(\tau))+1\leq 2.$$ Then, the Lavall\'ee-Poussin argument allows us to conclude that
$$\left\{\bigl[\log( X_0(\tau))\bigr]^+,\ \tau\in {\cal T}_T\right\}\ \ \ \ \mbox{is $Q$-uniformly integrable}.$$
Since $U^Q_0\left(t, x{\cal E}(\widehat\theta\cdot S)\right)$ is a martingale, we deduce that
$$\left\{\Bigl[U^Q_0\left(\tau, x{\cal E}_{\tau}(\theta\cdot S)\right)\Bigr]^+, \ \ \tau\in{\cal T}_T\right\}\ \ \ \ \ \mbox{is $Q$-uniformly integrable}.$$ A combination of this fact and the admissibility of $\theta$ leads to the uniform integrability of $U^Q_0\left(t, x{\cal E}_t(\theta\cdot S)\right)$. As a consequence, this process is a $Q$-supermartingale, and the proof of the theorem is completed.\qed

\begin{remark} It is clear from Theorem \ref{forwardlog} that if we assume that ${\widehat D}$ and $\overline D$ are predicable with finite variation, then $\widehat D$ is constant. Thus, in this case there is only one forward utility which is the classical log-utility augmented with a Hellinger process of an optimal pricing density. For the complete description of the optimal portfolio for the log-utility using the predictable characteristics, we refer the reader to G\"oll and Kallsen (2003). It is worth mentioning that our equation (\ref{martingaleofzero}) is the same equation that Kardaras (2012) derived in the one-dimensional context . The extension of this equation, called the "market equation",  to the multidimensional framework is attempted by Kabanov (2013) .
\end{remark}

%

\section{Discrete-time Market Models}\label{sec:examples}
This section illustrates our results of Sections \ref{powercase} and \ref{subsection:logcase} on the discrete-time market model, which is the most frequently used market models in the economic and/or financial literature.

 We consider the market model where trading times are $t=0,1,..,T$, the information flow of the market model is given by $\mathbb{F}=(\mathcal{F}_{n})_{n=0,1,...,T}$, and the $d$-dimensional stock price process is denoted by $S=(S_i)_{i=0,1,...,T}$. For $x\in (0,+\infty)$, $p<1$, and $t=0,1,...,T$, we put
\begin{equation}\label{Utility}
U_p(t,x):=\left\{\begin{array}{lll}D(t)x^p,\hskip 3cm\mbox{if}\ p\not =0,\\
\\
{\widehat D}(t)\log(x)+\overline{D}(t),\hskip 1cm\mbox{if}\ \ p=0,\end{array}\right.\end{equation}
where $D=(D(t))_{t=0,1,...,T}$, $({\widehat D}(t))_{t=0,1,...,T}$, and $(\overline{D}(t))_{t=0,1,...,T}$ are processes satisfying
\begin{equation}\label{condition}
\sup_{0\leq j\leq T}E\Bigl[\vert D(j)\vert+\vert {\widehat D}(j)\vert+\vert \overline{D}(j)\vert\Bigr]<+\infty.\end{equation}
Similarly as in (\ref{setD}) for the continuous-time case, we define ${\cal D}_j^+$  by
\begin{equation}\label{Dj}
{\cal D}_j^+:=\left\{\theta\in \mathbb R^d\Big|\ 1+\theta^{tr}x>0,\ G_j(dx)-a.e\right\},\ \ G_j(dx):=P(\Delta S_j\in dx\ |\ {\cal F}_{j-1}),\end{equation}
for any $j=1,2,...,T$, where $\Delta S_j:=S_j-S_{j-1}$.
For any process $X=(X_i)_{i=0,1,...,T}$, we associate to it the set of admissible strategies for the $j^{th}$ period of time ($j=1,2,...,T$) $\Theta_j(X,U_p)$, which is given by
\begin{equation}\label{AdmissibleTheta}
\Theta_j(X, U_p):=\Bigl\{\theta\in L^0({\cal F}_{j-1})\cap{\cal D}_j^+\ \Big|\ E\left(\vert X_j \Delta S_j\vert(1+\theta^{tr} \Delta S_j)^{p-1}\ \Big| {\cal F}_{j-1}\right)<+\infty,\ \ \ P-a.s.\Bigr\}.\end{equation}
In this framework, Assumption \ref{crucialassumption} becomes
\begin{assumption}\label{discretetime} For any  $j=1,2,...,T$, any $\theta\in {\mathcal{D}}_j^+$, $P$-a.e., and every sequence $(\theta_n)_{n\geq 1}\subset int(\mathcal{D}_j^+)$ that converges $P-a.s.$ to $\theta$, we have
\begin{eqnarray}
\label{power-assumption-dis}
&&\hskip -0.5cm\lim_{n\rightarrow +\infty}E\Big(|D(j)K_p(\theta_n^{tr}\Delta S_j)|\big|\mathcal{F}_{j-1}\Big)=\left\{
 \begin{array}{ll}
 +\infty, & \hbox{ on $\Gamma_j$;} \\
 E\Big(|D(j)K_p(\theta^{tr}\Delta S_j)|\big |\mathcal{F}_{j-1}\Big), & \hbox{ on $\Gamma_j^c$.}
 \end{array}
\right.
\end{eqnarray}
where $K_p(y):=y(1+y)^{1/(q-1)}=y(1+y)^{p-1}$ and $\Gamma_j:=\{G_j(\mathbb R^d)>0\ \ \mbox{and} \ \ \theta\not\in int(\mathcal{D}_j^+)\}$
\end{assumption}
 Below, we state our parametrization algorithm for forward utilities having the form of (\ref{Utility}).

\begin{theorem}\label{Discreteparametrizationpower}Let $p\in (-\infty,0)\cup(0,1)$.
Suppose that $S$ is bounded, {\bf Assumption \ref{discretetime}} holds, and $D$ satisfies (\ref{condition}). Then, the following are equivalent.\\
(i) $U_p(t,x)$ ---defined in (\ref{Utility})--- is a forward utility with the optimal portfolio $\widehat\theta=(\widehat\theta_i)_{i=1,2,...,T}$.\\
(ii) The two processes $D$ and $\widehat\theta$ are given by
\begin{eqnarray}
\widehat\theta_j\in\Theta_j(D,U_p)\ \ \mbox{is a root of}\ \ \ E\left(D_j \Delta S_j(1+\theta^{tr}\Delta S_j)^{p-1}\ \Big|\ {\cal F}_{j-1}\right)=0,\label{DThetaHat2}\\
\mbox{and}\ \ \ \ D(j-1)=E\left(D(j)(1+\widehat\theta^{tr}_{j}\Delta S_{j})^{p}\ \Big|\mathcal{F}_{j-1}\right),\label{DThetaHat1}\end{eqnarray}
for all $j=1,2,...,T$. Here $\Theta_j(D,U_p)$ is defined in (\ref{AdmissibleTheta}).\end{theorem}

\begin{remark} Theorem \ref{Discreteparametrizationpower} completely parameterizes the forward utilities of (\ref{Utility}) in the discrete time setting. In fact, the unique parameter for these forward utilities is the terminal value of the process $D$, which is $D(T)$. Given this random variable, we calculate the optimal portfolio for the $n^{th}$-period of time, $\widehat\theta_n$ as a root of equation (\ref{DThetaHat2}). Afterwards,  we calculate $D_{n-1}$ from (\ref{DThetaHat1}). Then, we repeat this procedure over and over again until we completely determine the two processes $D$ and $\widehat\theta$.
\end{remark}

{\bf Proof of Theorem \ref{Discreteparametrizationpower}}. Remark that, due to (\ref{condition}), the process $D$ can be represented by

$$\begin{array}{llll}
D(t)=D(0)Z^D_t\exp(a^D_t),\ \ \ \ \ \ \ Z^D_t:=\displaystyle\prod_{i=1}^t {{D(j)}\over{E(D(j)|{\cal F}_{j-1})}},\\
\\
a^D_t:=\displaystyle\sum_{j=1}^t \log\left[E\left({{D(j)}\over{D(j-1)}}\Big|{\cal F}_{j-1}\right)\right],\ \ \ \ \ \ t=1,2,...,T;\ \  Z^D_0=1,\ \ a^D_0=0.
\end{array}$$
It is easy to verify that $Z^D$ is a positive martingale (since $pD(j)>0$) and $a^D$ is predictable. Thus, throughout the remaining part of the proof, we will consider the probability measure $Q:=Z^D_T\cdot P$. We will start by proving $(i)\Longrightarrow (ii)$. Thus, suppose that $(i)$ holds. Then there exists an admissible strategy $\widehat\theta$ such that for any other admissible strategy $\theta$, the processes $U_p\left(j,\displaystyle\prod_{k=1}^j(1+\widehat\theta_k^{tr}\Delta S_k)\right)$ and $U_p\left(j,\displaystyle\prod_{k=1}^j(1+\theta_k^{tr}\Delta S_k)\right)$ are martingale and supermartingale respectively. This implies that for any $j=1,2,...,T$,
\begin{equation}\label{inequlityfordiscrete}
D(0) E^Q\left((1+\theta_j^{tr}\Delta S_j)^p\Big|{\cal F}_{j-1}\right)\leq D(0)e^{-a^D_j+a^D_{j-1}}=D(0) E^Q\left((1+\widehat\theta_j^{tr}\Delta S_j)^p\Big|{\cal F}_{j-1}\right).\end{equation}
Then, the equality in the RHS term of (\ref{inequlityfordiscrete}) implies (\ref{DThetaHat1}). While the whole inequality (\ref{inequlityfordiscrete}) can be transformed into
$$
D(0)\int (1+\theta_j^{tr} x)^p G_j^Q(dx)\leq D(0)\int (1+\widehat\theta_j^{tr} x)^p G_j^Q(dx),$$
where $G_j^Q(dx)$ is give by $ G_j^Q(dx):=Q\left(\Delta S_j\in dx\Big|\ {\cal F}_{j-1}\right)$. Due to Lemma \ref{lem-interior-integrability} and {\bf Assumption \ref{discretetime}}, we conclude that
\begin{equation}\label{dis-power-psi}
\Psi_j(\lambda):=D(0)\int (1+\lambda^{tr} x)^p G_j^Q(dx),\ \ \ \lambda\in {\cal D}_j^+,\end{equation} is differentiable on int$({\cal D}_j^+)$, and its minimum $\widehat\theta_j$ belongs to int$({\cal D}_j^+)$ and is a root for
$$
0=\nabla \Psi_j(\lambda)=pD(0)\int (1+\lambda^{tr} x)^{p-1}x G_j^Q(dx).$$
This is equivalent to (\ref{DThetaHat2}), and assertion (ii) follows.\\
\noindent To prove the reverse (i.e. $(ii)\Longrightarrow (i)$), we suppose that assertion (ii) holds. Then by multiplying both sides of (\ref{DThetaHat1})
by $x^p\prod_{k=1}^{j-1}(1+\widehat\theta^{tr}_j\Delta S_k)^p$, we obtain
$$
 D(j-1)x^p\prod_{k=1}^{j-1}(1+\widehat\theta^{tr}_k\Delta S_k)=E\left(D(j)x^p\prod_{k=1}^{j}(1+\widehat\theta^{tr}_j\Delta S_k)^p\ \Big|\mathcal{F}_{j-1}\right).$$
 This proves that for any $x\in (0,+\infty)$ the process $U_p\left(j, x\displaystyle\prod_{k=1}^{j}(1+\widehat\theta^{tr}_k\Delta S_k)\right),\ \  j=0,1,...,T,$ is a martingale. Since $pD(j)>0$ and $p<1$ for any $j=0,1,...,T$, then for any admissible portfolio rate $\theta$, we derive
 $$
 D(j)(1+\theta^{tr}_j\Delta S_j)^p-D(j)(1+\widehat\theta^{tr}_j\Delta S_j)^p\leq pD(j)\left(\theta_j-\widehat\theta_j\right)^{tr}\Delta S_j(1+\widehat\theta^{tr}_j\Delta S_j)^{p-1}.$$
 Then, by taking conditional expectation in both sides of the above inequality and using (\ref{DThetaHat2}) and (\ref{DThetaHat1}), we obtain
 $$
 E\left(x^p D(j)(1+\theta^{tr}_j\Delta S_j)^p\Big|{\cal F}_{j-1}\right)\leq E\left(x^p D(j)(1+\widehat\theta^{tr}_j\Delta S_j)^p\Big|{\cal F}_{j-1}\right)=x^p D(j-1).$$
 Then by multiplying both sides of this inequality with $\displaystyle\prod_{k=1}^{j-1}(1+\theta^{tr}_k\Delta S_k)^p$, we conclude that
 $$
U_p\left(j, x\displaystyle\prod_{k=1}^{j}(1+\theta^{tr}_j\Delta S_k)\right)=x^p D(j)\prod_{k=1}^{j}(1+\theta^{tr}_k\Delta S_k)^p,\ \ \ j=0,1,...,T$$
is a supermartingale for any $x>0$ and any admissible $\theta$. This ends the proof of the theorem.\qed
\vskip .5cm
One of the easiest and popular case of discrete-time market model is the binomial model. Let $\xi_j$ be a $\mathcal{F}_j$-measurable random variable, which takes only two values, $\xi^u_{j}$ and $\xi^d_{j}$ satisfying $0<\xi_{j}^d<1<\xi_{j}^u$ for any $j=1,2,...,T$. Given the price of the stock at time $j-1$ (i.e. $S_{j-1}$), the price at time $j$ will either go up  to $S_{j-1}\xi^u_{j}
$ or go down to $S_{j-1}\xi_{j}^d
$. Therefore, we get
$$S_{j}=S_{j-1}\xi_{j}=S_0\prod_{k=1}^{j}\xi_k,\ \ \ \ \ \ S_0>0.$$
We denote by $(A_j)_{j=1,...,T}$ the sequence of events given by
\begin{equation}\label{event-dis-A}
A_{j}:=\{\xi_{j}=\xi^u_{j}\}\in\mathcal{F}_{j}.\end{equation}
For this model, we have $\#(\Omega)=2^N<+\infty.$  Thus, any random variable is integrable, and
$$
\Theta_j(D,U_p)=L^0(\mathcal{F}_{j-1})\cap \mathcal{D}_j^+,\quad j=1,2,...,T.$$
Furthermore, in this case Assumption \ref{discretetime} is always fulfilled due to
\begin{equation}\label{dis-powerBN-D1}
\mathcal{D}_{j}^+=\Big]1/{(1-\xi^u_j)S_{j-1}},\ \ 1/{(1-\xi^d_j)S_{j-1}}\Big[=int(\mathcal{D}_j^+),\ \ \ \ \ \forall\ \ j=1,2,...,T.
\end{equation}
 The description of the power-type forward utilities in this simple framework ---generalizes the results of Musiela and Zariphoupoulou (2009a) to the power case--- has a more explicit form and is given in the following.

\begin{corollary}\label{disBN-power-paraThem}
 The following two assertions are
equivalent.\\
(i) $U_{p}(t,x)$, defined in (\ref{Utility}), is a forward utility with the optimal
portfolio rate $\widehat\theta=(\widehat\theta_{j})_{j=1,2,...,T}$.\\
(ii) The process $D$ is a supermartingale having the multiplication Doob-Meyer decomposition, $D=D(0) Z^D\exp(a^D)$  ($Z^D$ is a positive martingale and $a^D$ is predictable) such that the following properties hold:\\
(ii.1) By putting $Q:=\left(Z^D_T/Z^D_0\right)\cdot P$, then for $j\in\{1,2,...,T\}$, $\widehat\theta_j$ is given by
\begin{equation}\label{dis-powerBN-thetahat}
\widehat\theta_j=\frac{\gamma_j-1}{(\xi^u_j-1-\gamma_j\xi^d_j+\gamma_j)S_{j-1}}\in\mathcal{D}_j^+,\ \gamma_j:=\left(\frac{(\xi^u_j-1)Q(A_{j}|\mathcal{F}_{j-1})}{(1-\xi^d_j)Q(A^c_{j}|\mathcal{F}_{j-1})}\right)^{1-q}
\end{equation}
(ii.2) The predictable process $a^{D}$ is given by
\begin{eqnarray*}
a^{D}_{j}=-\sum_{k=1}^j\log\Big(\frac{(\gamma_k^{p-1}Q(A_{k}|\mathcal{F}_{k-1})+Q(A^c_{k}|\mathcal{F}_{k-1}))
(\xi^u_k-\xi^d_k)^{p-1}}{(\xi^u_k-1-\gamma_k\xi^d_k+\gamma_k)^{p-1}}\Big),\ \ \ \ \ j=1,2,...,T.
\end{eqnarray*}
\end{corollary}

{\bf Proof.}
This corollary can be obtained as an application of Theorem \ref{Discreteparametrizationpower}. Thus, we will avoid to repeat the same proof again by giving some remarks emphasizing the nice features of this case that simplify tremendously the proof. Since $\mathcal{D}_{j}^+$ is open and $\#(\Omega)<+\infty$, the assumptions (\ref{condition}) and (\ref{power-assumption-dis}) are automatically fulfilled. The function $\Psi_j$ given by (\ref{dis-power-psi}) becomes
$$\Psi_j(\lambda)= Q(A_{j}|\mathcal{F}_{j-1}){(1+(\xi^u_{j}-1)\lambda S_{j-1})^{p}}+ Q(A^c_{j}|\mathcal{F}_{j-1}){(1+(\xi^d_{j}-1)\lambda S_{j-1})^{p}},$$
which is differentiable on $\mathcal{D}_j^+$.
Thus, $\widehat\theta_j$ is the solution of the equation, $\Psi'(\lambda)=0$, which leads to (\ref{dis-powerBN-thetahat}). Finally, $a^D_j$ is derived by plugging (\ref{dis-powerBN-thetahat}) into (\ref{DThetaHat1}) and applying the decomposition of $D$.
This ends the proof of the corollary.
\qed
\vskip .5cm
 We conclude this subsection by deriving the results for the logarithm case.

\begin{theorem}\label{Discreteparametrizationlog}
Suppose that $S$ is bounded, ${\widehat D}$ and $\overline{D}$ two processes satisfy (\ref{condition}), and Assumption \ref{discretetime} for $M_t\equiv E\left({\widehat D}_T\ \big|\ {\cal F}_t\right)/E{\widehat D}(T)$ holds. Then, the following are equivalent.\\
(i) The functional $U_0(t,x):={\widehat D}(t)\log(x)+\overline{D}(t)$ (see (\ref{Utility})), is a forward utility with the optimal portfolio $\widehat\theta=(\widehat\theta_i)_{i=1,2,...,T}$.\\
(ii) ${\widehat D}$ is a positive martingale, the process $\widehat\theta$ satisfies
\begin{equation}\label{DThetaHat2log}
\widehat\theta_j\in\Theta_j(\widehat D, U_0)\ \ \mbox{and $\widehat\theta_j$ is a root for}\ \ \ E\left({{{\widehat D}(j)}\over{1+\theta^{tr}\Delta S_j}} \Delta S_j\ \Big|\ {\cal F}_{j-1}\right)=0,\ \ \ \ \ j=1,2,...,T,\end{equation}
and the process $\overline{D}$ is a supermartingale with predictable part given by
\begin{equation}\label{DThetaHat1log}
-\sum_{k=1}^j E\left({\widehat D}(k)\log(1+\widehat\theta^{tr}_{k}\Delta S_{k})\ \Big|\mathcal{F}_{k-1}\right).\end{equation}
 Here $\Theta_j(\widehat D,U_0)$ is given by (\ref{AdmissibleTheta}).\end{theorem}

{\bf Proof.} Suppose that assertion (i) holds. Then, for any $x\in (0,+\infty)$, the process
$$U_0\left(j, x\prod_{k=1}^j(1+\widehat\theta_k^{tr}\Delta S_k)\right)\ \ \ \ \mbox{is a martingale}.$$
Then, we deduce that both processes
$$
\overline{D}(j)+{\widehat D}(j)\log\left(\prod_{k=1}^j(1+\widehat\theta_k^{tr}\Delta S_k)\right),\ \ \mbox{and}\ \ {\widehat D}(j),$$ are martingales. This proves that ${\widehat D}$ is a positive martingale. Since for any admissible $\theta$,
$$U_0\left(j, x\prod_{k=1}^j(1+\theta_k^{tr}\Delta S_k)\right),$$
is a supermartingale. Then  we derive
\begin{equation}\label{inequalityStar}
\begin{array}{lll}
E\left({\widehat D}(j)\log(1+\theta_j^{tr}\Delta S_j)\Big|{\cal F}_{j-1}\right)\leq E\left(\overline{D}(j-1)-\overline{D}(j)\Big|{\cal F}_{j-1}\right)\\
\\
\hskip 5.5cm =E\left({\widehat D}(j)\log(1+\widehat\theta_j^{tr}\Delta S_j)\Big|{\cal F}_{j-1}\right).\end{array}\end{equation}
As a result, the equality in the RHS term of (\ref{inequalityStar}) implies that
$$
A^{\overline{D}}_j:=\sum_{k=1}^j E\left(\overline{D}(k)-\overline{D}(k-1)\Big|{\cal F}_{k-1}\right)=-\sum_{k=1}^j E\left({\widehat D}(k)\log(1+\widehat\theta_k^{tr}\Delta S_k)\Big|{\cal F}_{k-1}\right).$$ This proves (\ref{DThetaHat1log}). If we put $Q:={\widehat D}(T)\left({\widehat D}(0)\right)^{-1}\cdot P$, then (\ref{inequalityStar}) becomes
$$
\int \log(1+\theta_j^{tr} x)G_j^Q(dx)\leq \int \log(1+\widehat\theta_j^{tr} x)G_j^Q(dx),$$ where $G^Q_j(dx):=Q(\Delta S_j\in dx|{\cal F}_{j-1})$. This proves that $\widehat\theta_j$ maximizes the function
\begin{equation}\label{dis-log-phi}
{\cal Y}_j(\lambda):= \int \log(1+\lambda^{tr} x)G_j^Q(dx), \end{equation} on the set ${\cal D}_j^+$. Due to Lemma \ref{lem-interior-integrability} , we conclude that $\widehat\theta_j\in \mbox{int}\left({\cal D}_j^+\right)$, and $\widehat\theta_j$ is a root for
$$\begin{array}{lll}
0={\widehat D}(j-1)\nabla{\cal Y}_j(\lambda)=\displaystyle {\widehat D}(j-1)\int \left(1+\lambda^{tr} x\right)^{-1}xG_j^Q(dx)\\
\\
\hskip 3.5cm =E\left({\widehat D}(j)(1+\lambda^{tr} \Delta S_j)^{-1}\Delta S_j\Big|{\cal F}_{j-1}\right).\end{array}$$
This proves (\ref{DThetaHat2log}), and the proof of assertion (ii) is completed.\\
To prove that (ii) implies (i), we assume that assertion (ii) holds. Then, due to (\ref{DThetaHat2log}), and the fact that ${\widehat D}$ is a martingale, we calculate
$$\begin{array}{llll}
E\left(U_0(j, x\prod_{k=1}^j (1+\widehat\theta^{tr}_{k}\Delta S_{k}))\Big|{\cal F}_{j-1}\right)= {\widehat D}(j-1)\log(x)+E\left(\overline{D}(j)\Big|{\cal F}_{j-1}\right)+\\
\\
\hskip 1.5cm {\widehat D}(j-1)\displaystyle \sum_{k=1}^{j-1}\log(1+\widehat\theta^{tr}_{k}\Delta S_{k})+E\left({\widehat D}(j)\log(1+\widehat\theta^{tr}_{j}\Delta S_{j})\Big|{\cal F}_{j-1}\right),\\
\\
\hskip 6cm =\displaystyle U_0\left(j-1, x\prod_{k=1}^{j-1} (1+\widehat\theta^{tr}_{k}\Delta S_{k})\right).
\end{array}$$
The last equality follows easily from (\ref{DThetaHat1log}). This proves that $\displaystyle U_0\left(j, x\prod_{k=1}^j (1+\widehat\theta^{tr}_{k}\Delta S_{k})\right)$ is a martingale for any $x\in (0,+\infty)$. Thanks to (\ref{DThetaHat2log}) and the concavity of the log function, we obtain for any admissible $\theta$
$$
E\left({\widehat D}(j)\left[\log(1+\theta^{tr}_{j}\Delta S_{j})-\log(1+\widehat\theta^{tr}_{j}\Delta S_{j})\right]\Big|{\cal F}_{j-1}\right)\leq 0.$$
Then, by combining this inequality  with
$$\begin{array}{llll}
U_0\left(j, x\displaystyle\prod_{k=1}^j (1+\theta^{tr}_{k}\Delta S_{k})\right)= U_0\left(j, x\displaystyle\prod_{k=1}^j (1+\widehat\theta^{tr}_{k}\Delta S_{k})\right)+
{\widehat D}(j)\displaystyle \sum_{k=1}^{j}\log\left({{1+\theta^{tr}_{k}\Delta S_{k}}\over{1+\widehat\theta^{tr}_{k}\Delta S_{k}}}\right),
\end{array}$$
we conclude that the process $U_0\left(j, x\displaystyle\prod_{k=1}^j (1+\theta^{tr}_{k}\Delta S_{k})\right)$ is a supermartingale. This completes the proof of the theorem.
\qed\\

For the binomial model,  the sets $\Theta_{j}(\widehat D, U_0)$ and $\mathcal{D}_j^+$ are similar to those of the power case. The characterization of the logarithm forward utilities in binomial model is stated as follows.
\begin{corollary}\label{disBN-log-paraThem}
The following two assertions are
equivalent.\\
(i) $U_{0}(t,x)$ defined in (\ref{Utility}), is a forward utility with optimal
portfolio $\widehat\theta=(\widehat\theta_{j})_{j=1,2,...,T}$.\\
(ii) The following properties hold:\\
(ii.1) $D_{0}$ is a positive martingale and
$\widehat\theta_j$ is given by
\begin{equation}\label{dis-logBN-thetahat}
\widehat\theta_j=\frac{(\xi^u_j-1)Q_j-(1-\xi^d_j)(1-Q_j)}{(\xi^u_j-1)(1-\xi^d_j)S_{j-1}}\in\mathcal{D}_j^+,
\end{equation} where $Q_j:=Q(A_{j}|\mathcal{F}_{j-1})$, $Q:=\frac{D_{0}(T)}{D_{0}(0)}\cdot P$ and  $A_j$ is given by (\ref{event-dis-A}).\\
(ii.2)  $\overline{D}$ is a supermartingale with  predictable part  given by
\begin{equation}\label{dis-BN-log-pre}
-\sum_{k=1}^j\Big[\log\Big(\frac{\xi^u_k-\xi^d_j}{1-\xi^d_k}Q_j\Big)Q_j+\log\Big(\frac{\xi^u_k-\xi^d_j}{\xi^u_k-1}(1-Q_j)\Big)(1-Q_j)\Big].
\end{equation}
\end{corollary}

{\bf Proof.}
The proof of this corollary follows from Theorem \ref{Discreteparametrizationlog}, and the fact that the function ${\cal Y}_j(\theta)$ defined in (\ref{dis-log-phi}) takes the form of

 \begin{equation}\label{log-dis-phi}
 {\cal Y}_j(\theta)=\log \left((1+(\xi^u_{j}-1)\theta S_{j-1})\right)Q_j+\log \left((1+(\xi^d_{j}-1)\theta S_{j-1})\right)
(1-Q_j),
\end{equation}
in the binomial context. \qed
\begin{remark}\label{rem-muti-binomial}
An extension of the binomial discrete model is the multi-dimensional discrete model, where the $d$-dimensional stock price process branches into $n$ ($n> 2$) possible values at any time. For such model, Assumption \ref{discretetime} ---for the case of $\widehat D$ and $p=0$ instead of $D$---  and (\ref{condition}) are automatically satisfied since the set $\mathcal{D}^+$ is open and $\#(\Omega)<+\infty$. The characterization of HARA forward utilities for this model and other illustrative examples can be found in Ma (2013).
\end{remark}


\section{Conclusion and Related Open problems}
This paper describes completely and explicitly the HARA-type forward utilities when $S$ is locally bounded and under some mild assumptions (see {\bf Assumption \ref{crucialassumption}}) on the parameters of these random field utilities and those of the market model. There is no doubt that our approach is the most explicit that deals with the most general market models. However, there are some points where our results can be improved and/or extended. Furthermore, our work leads to a number of related open problems. Below, we outline some of these problems and possible extensions.\\

\noindent(1) Our results can be extended to the case of general $S$ (not necessarily locally bounded) on the one hand. On the other hand, {\bf Assumption \ref{crucialassumption}} can be ignored, and this leaded to the new concept of minimal Hellinger deflator. This last part constitutes our current work ---that is in progress--- in Choulli and Ma (2013), where the explicit forms for the forward utilities and their optimal portfolios are lost and the characterization is achieved by duality only. Also, we can conjecture that the non-arbitrage assumption of ${\cal Z}_{loc}^e(P)\not=\emptyset$ is redundant when $U$ is a forward utility. This can be proved by extending the results of Choulli, Deng and Ma (2013) to the case of random field utilities.\\

\noindent(2) Can we define a {\it Forward "Regularis\'ee"} for any random field utility?  This forward regularis\'ee will be the smallest forward utility that is bigger or equal to the random field utilities.\\

\noindent(3) The parametrization of forward utilities in discrete time setting is reduced to one parameter which is the terminal value of the process $D$. In other words, the forward utility and its optimal portfolio as well, in  that context, are calculated using backward iterations. This leads to the question whether we can characterize these forward utilities, in the continuous-time framework, using backward stochastic differential equations (BSDEs). This sounds a very promising and interesting alternative to the approach of PDEs. For the exponential case, we learned from Christoph Frei that, this question was addressed recently by Anthropelos (2013).\\

\noindent(4) An other problem is how these forward utilities and their optimal portfolios are altered in defaultable markets, or more generally by any random exit time. We believe that this question will lead to the stability of these forward utilities under uncertainty models. As a consequence, we might be able to explain the interplay between the uncertainty models and models with random horizons.\\

\appendix

\begin{center}
    \Large \textbf{Appendix}
\end{center}
This appendix contains two sections.

\section{Some Useful Intermediatory Results}

\noindent For the  following representation theorem, we refer to
Jacod (1979) (Theorem 3.75, page 103) and to Jacod and Shiryaev (2003) (Lemma 4.24,
page 185).

\begin{theorem}\label{representation}
Let $N\in {\cal M}_{0,loc}$. Then, there exist a predictable and
$S^c$-integrable process $\phi$, $N'\in {\cal M}_{0,loc}$ with
$[N',S]=0$ and functionals $f\in {\widetilde{{\cal P}}}$ and $g\in
{\widetilde{{\cal O}}}$ such that\\
(i)\hskip 0.25cm $
 \Bigl (\displaystyle\sum_{s=0}^t f(s, \Delta S_s )^2
I_{\{\Delta S_s\not = 0\}}\Bigr )^{{1\over{2}}}\in{{\cal A}}^+_{loc},\
 \Bigl (\sum_{s=0}^t g(s, \Delta S_s )^2
I_{\{\Delta S_s\not = 0\}}\Bigr )^{{1\over{2}}}\in{{\cal A}}^+_{loc}, $\\
(ii)\hskip 0.25cm $
M^P_{\mu}(g\ |\ {\widetilde {{\cal P}}})=0,$\\
(iii)\hskip 0.2cm The process $N$ is given by \begin{equation}
\label{Ndecomposition}
 N=\phi\cdot
S^c+W\star(\mu-\nu)+g\star\mu+{N'},\quad W=f+\frac{{\widehat
f}}{1-a}I_{\{a<1\}}.
\end{equation} Here $\widehat f_t=\displaystyle\int f_t(x)\nu(\{t\},d x)$ and $f$ has a version
such that $\{a=1\}\subset \{\widehat f=0\}$.\\ Moreover
 \begin{equation}
 \label{jumps}
 \Delta N_t= \Bigl(f_t(\Delta S_t)+g_t(\Delta S_t)\Bigr)I_{\{\Delta
 S_t\not=0\}}-{{\widehat
f_t}\over{1-a_t}}I_{\{\Delta
 S_t=0\}}+\Delta N_t'.
 \end{equation}
\end{theorem}

The quadruplet $ (\beta, f, g, N') $ is called throughout the paper by Jacod's components/parameters of $N$ (under $P$).\\

For any probability measure $R$, which is equivalent to $P$, we denote by $(b^R,c,\nu^R)$ (where $\nu^R(dt,dx)=F_t^RdA_t$) the predictable characteristics of $S$ with respect to $R$. Remark that since $R\sim P$, we have $F^R\sim F$, $P\otimes A-$a.e. As a result, the sets ${\cal D}^+$  and ${\cal D}_1$ ---defined in (\ref{setD}) and (\ref{power-d1}) respectively--- do not depend on equivalent probabilities.
\begin{lemma}\label{DintD}
Suppose that $S$ is locally bounded. Let $R$ be a probability measure equivalent to $P$. Then, the interior of $\mathcal{D}^+$ satisfies
\begin{equation}\label{power-d1}
0\in int(\mathcal{D}^+)=\mathcal{D}_1:=\{\lambda\in\mathcal{D}^+:\ \exists\ \ \delta>0, 1+\lambda^{tr}x\geq\delta,\ \ F-a.e.\}.\end{equation}
\end{lemma}
\begin{proof} By stopping, there is no loss of generality in assuming that $S$ is bounded by $K$. For any $\lambda_0\in int(\mathcal{D}^+)$, there exists $\varepsilon>0$ such that for any $B(\lambda_0,\epsilon)\subset{\cal D}^+$. Remark that
 $\lambda:=\lambda_0(1+{{\epsilon}\over{2\vert\lambda_0\vert}})\in B(\lambda_0,\epsilon)$, hence we have $1+\lambda^{tr} x>0,\ \ F^R-a.e.$ Therefore, we write
 $$
 1+\lambda_0^{tr} x={{\epsilon}\over{\epsilon+2\vert\lambda_0\vert}}+{{2\vert\lambda_0\vert(1+\lambda^{tr} x)}\over{\epsilon+2\vert\lambda_0\vert}}\geq {{\epsilon}\over{\epsilon+2\vert\lambda_0\vert}}>0,\ \ \ F^R-a.e$$
 This proves that $\lambda_0\in {\cal D}_1$ and hence we get $int({\cal D}^+)\subset {\cal D}_1$. In the remaining part fo the proof we will prove the reverse inclusion. Let $\lambda_0\in {\cal D}_1$ with uniform bound from below $\delta$. Then, it is easy to see that $B(\lambda_0,\delta/K)\subset{\cal D}^+$. This proves that $\lambda_0\in int({\cal D}^+)$, and the proof of the lemma follows immediately from noticing that $0\in {\cal D}_1$.
\end{proof}

\begin{lemma}\label{lem-interior-integrability}
Suppose $S$ is locally bounded and let $R$ be a probability measure equivalent to $P$. Then, the following assertions hold, $P\otimes A$-a.e.\\
(i) For any $\lambda\in int(\mathcal{D}^+)$,
\begin{equation}\label{interior-int}
\int|x|\big|(1+\lambda^{tr}x)^{1/(q-1)}-1\big|F^R(dx)<+\infty.
\end{equation}
(ii) The function $\Phi^R_p(\lambda)$ defined in (\ref{Phip}) is convex, proper, lower semi-continuous, differentiable on $int(\mathcal{D}^+)$, and
$$\nabla\Phi^R_p(\lambda_0)=b^R+\frac{c\lambda_0}{q-1}+\int \left[x(1+\lambda_0^{tr}x)^{1/(q-1)}-x\right]F^R(dx),\ \ \ \ \ \forall\ \ \  \lambda_0\in int(\mathcal{D}^+).$$
(iii) If $\displaystyle\min_{\lambda\in{\cal D}}\Phi_p^R(\lambda)=\Phi_p^R(\lambda_0)$ for $\lambda_0\in\mathbb R^d$, and {\bf Assumption \ref{crucialassumption}} with $M_t\equiv E\left({{dR}\over{dP}}\big|\ {\cal F}_t\right)$ holds, then $\lambda_0\in int({\cal D}^+)$.
\end{lemma}
{\bf Proof.} The proof of this lemma will be achieved in three parts namely parts a), b) and c) where we will prove assertions (i), (ii) and (ii) respectively.\\
{\bf a)} For any $\lambda\in int(\mathcal{D}^+)$,
due to Lemma \ref{DintD}, there exists $\delta\in(0,1)$ such that $1+\lambda^{tr}x\geq\delta>0$, $F-a.e$. Then, an application of Taylor's expansion to $(1+\lambda^{tr}x)^{1/(q-1)}-1=(1+\lambda^{tr}x)^{p-1}-1$ leads to the existence of $ r\in(0,1)$ such that
$$
\vert(1+\lambda^{tr}x)^{p-1}-1\vert=(1-p)\vert\lambda^{tr}x\vert(1+r\lambda^{tr}x)^{p-2}\leq \delta^{p-2}(1-p)\vert\lambda^{tr} x\vert.
$$
By combining this  $\int |x|^2F^R(dx)<+\infty,\quad P\otimes A-a.s.$ (that follows from the fact that $S$ is locally bounded), we derive
$$
\int|x|\big|(1+\lambda^{tr}x)^{1/(q-1)}-1\big|F^R(dx)\leq \delta^{p-2}|\lambda|(1-p)\int |x|^2F^R(dx)<+\infty,\ \ \ \ P\otimes A-a.e.
$$ This proves assertion (i) of the lemma.\\
\noindent {\bf b)} It is obvious that $\Phi_p^R$ is convex, proper and semi-lower continuous due to Fatou's Lemma and $f_p\geq 0$. Then, the proof of assertion (ii) of the lemma will be completed once we prove the differentiability of $\Phi_p^R$ on $int({\cal D}^+)$. Let $\lambda_0\in int(\mathcal{D}^+)$. Then, for any $y\in\mathbb R^d$, thanks to assertion (i) proved in part a) and Lemma \ref{DintD}, there exists $\varepsilon_0>0$ such that for any $0\leq \varepsilon\leq\varepsilon_0$, $\lambda_0+\varepsilon y\in dom(\Phi^R_p)$. \\
An application of Taylor's expansion of the function $g_p(\lambda^{tr}x):=\frac{(1+\lambda^{tr}x)^p-1-p\lambda^{tr}x}{q(q-1)}$ implies the existence of $r\in(0,1) $ such that
$$k_{\varepsilon}(x):=\frac{g_p(\lambda_0^{tr}x+\varepsilon y^{tr}x)-g_p(\lambda_0^{tr}x)}{\varepsilon}=y^Tx\left((1+\lambda_0^Tx+r\varepsilon y^Tx)^{1/(q-1)}-1\right).$$
Meanwhile, notice that $(|k_{\varepsilon}(x)|)_{\varepsilon}$ is bounded from above by $$k(x):=|y||x|\max\Big(\big|(1+\lambda_0^{tr}x)^{1/(q-1)}-1\big|,\big|(1+\lambda_0^{tr}x+\varepsilon_0 y^{tr}x)^{1/(q-1)}-1\big|\Big).$$
Thanks to Lemma \ref{lem-interior-integrability}--(i), $k(x)$ is $F$-integrable due to $\lambda_0,\lambda_0+\varepsilon_0 y \in dom(\Phi^R_p)$.
It allows us to apply Dominated Convergence Theorem to $(|k_{\varepsilon}(x)|)_{\varepsilon}$, which leads to
\begin{equation}\label{differnetial}\lim_{\varepsilon\rightarrow 0}\frac{\Phi^R_p(\lambda_0+\varepsilon y)-\Phi^R_p(\lambda_0)}{\varepsilon}=y^{tr} \Phi^*(\lambda_0),\end{equation}
where $\Phi^*(\lambda_0)$ is given by
$$
\Phi^*(\lambda_0):=b^R+\frac{c\lambda_0}{q-1}+\int \left[x(1+\lambda_0^{tr}x)^{1/(q-1)}-x\right]F^R(dx).$$
It is clear from (\ref{differnetial}) that $y^{tr}\Phi^*(\lambda_0)$ is the directional derivative of $\Phi^R_p$ at $\lambda_0$, which is linear in $y$. Thus, due to Theorem 25.2 in Rockafellar (1970), $\Phi^R_p$ is differentiable at $\lambda_0\in int(\mathcal{D}^+)$.\\
{\bf c)} To prove assertion (iii),  we start by noticing that $\lambda_0$ belongs ${\cal D}^+$ ---since ${\cal D}^+\subset dom(\Phi_p^R)$---, and $\lambda_n:={{n-1}\over{n}}\lambda_0\in int({\cal D}^+)={\cal D}_1$. Thus, a direct application of assertion (ii) to each $\lambda_n$, we deduce that $\Phi_p^R$ is differentiable at $\lambda_n$, and due to the convexity of $\Phi_p^R$, we get
$$
\lambda_n^{tr} \nabla\Phi^R_p(\lambda_n)\leq (n-1)\left[\Phi_p^R(\lambda_0)-\Phi^R_p(\lambda_n)\right]\leq 0.$$
This implies that
$$
\int \lambda_n^{tr} x{{(1+\lambda_n^{tr} x)^{p-1}-1}\over{p-1}}F^M(dx)\leq -\lambda_n^{tr} b+(1-p)\lambda_n^{tr} c\lambda_n.$$ Thus a combination of this with assumption \ref{crucialassumption} imply that $\lambda_0\in int({\cal D}^+)$ (otherwise, by using Fatou's lemma in the above inequality, we will obtain $+\infty\leq \lambda_0^{tr}b+(1-p)\lambda_0^{tr}c\lambda_0$ which is impossible). This completes the proof of the lemma.
\qed
%

\begin{proposition}\label{theta-Doobdecomposition}
Suppose that $S$ is locally bounded with the following decomposition
$$
S=S_0+S^c+z\star(\mu-\nu)+b\cdot A.$$
Let $\theta$ be an $S$-integrable process, and $\alpha\in (0,+\infty)$. Then the following assertions hold.\\
(i) The process
\begin{equation}\label{Xtheta}
X^{\theta}:=\theta\cdot S-\sum \theta^{tr}\Delta S I_{\{\vert \theta^{tr} \Delta S\vert>\alpha\}},\end{equation}
is a locally bounded semi-martingale.\\
(ii) If we denote
$$\xi^{\theta}:=\theta^{tr} b-\int (\theta^T z) I_{\{\vert \theta^{tr} z\vert>\alpha\}}F(dz),$$
then $\vert \xi^{\theta}\vert\cdot A\in {\cal A}^+_{loc}$.\\
(iii) The process
$
X^{\theta}-\xi^{\theta}\cdot A,$
is a local martingale.
\end{proposition}
{\bf Proof.} The proof of assertion (i) is classic, and can be found in Dellacherie-Meyer(1980) or Jacod and Shiryaev (2003). Now, we will focus on proving simultaneously the remaining assertions.\\
Since $S$ is locally bounded, then it is clear that $\theta I_{\{\vert \theta\vert\leq n\}}\cdot S$  and $I_{\{\vert \theta\vert\leq n\}}\cdot X^{\theta}$ are locally bounded semimartingales. Therefore, $\sum \theta^{tr}\Delta S I_{\{\vert \theta^{tr} \Delta S\vert>\alpha,\ \vert \theta\vert\leq n\}}$ is a locally bounded process with finite variation, and its compensator is given by
$$
V^{\theta,n}:=(\theta^{tr} z) I_{\{\vert \theta^{tr} z\vert>\alpha,\ \vert \theta\vert\leq n\}}\star\nu.
$$
It is obvious that the two processes
$$
 \theta^T I_{\{\vert \theta\vert\leq n\}}\cdot S-\theta^{tr} b I_{\{\vert \theta\vert\leq n\}}\cdot A,\ \ \mbox{and}\ \
 \sum \theta^{tr}\Delta S I_{\{\vert \theta^{tr}\Delta S\vert>\alpha,\ \vert \theta\vert\leq n\}}-V^{\theta,n}$$
are local martingales. Since $\widetilde X^{\theta}$ ---the compensator of $X^{\theta}$--- exists and is a locally integrable process, then we derive
$$
Var\left({\widetilde X}^{\theta}\right)=\lim_{n} Var\left(I_{\{ \vert\theta\vert\leq n\}}\cdot {\widetilde X}^{\theta}\right)=\vert \xi^{\theta}\vert\cdot A.$$ this proves both assertions (ii) and (iii). Here, for any process with finite variation, $V$, we denote its variation by $Var(V)$. \qed

\begin{proposition}\label{theta-supermatingale}
Let $p\in (-\infty,0)\cup(0,1)$, $\widetilde{Z}$ be a martingale density, and $\widehat{\theta}\in \Theta(1,U_p)$ such that $\widehat{Z}:=\widetilde{Z}{\cal E} (\widehat{\theta}\cdot S)$ is a true martingale. If we denote $\widehat{Q}:=\widehat{Z}_{T}\cdot P$ and consider $\theta\in\Theta(1, U_p)$ satisfying
\begin{equation}\label{p<0-integrability}
\sup_{\tau\in{\cal T}_{T}}E^{\widehat{Q}}\left({\cal E}_{\tau}(\theta\cdot S)^{p}{\cal E}_{\tau}(\widehat{\theta}\cdot S)^{-p}\right)<+\infty,
\end{equation}
then $\displaystyle\mathrm{sign}(p){\cal E}(\theta\cdot S)^p{\cal E}(\widehat{\theta}\cdot S)^{-p}$ is a $\widehat{Q}$-supermartingale.
\end{proposition}

{\bf Proof.} Notice that the case of $p\in (0,1)$, the proposition is trivial and (\ref{p<0-integrability}) is always true.\\ In the remaining part of the proof we assume that $p<0$ and we consider $(T_{n})_{n\geq 1}$ a sequence of stopping times that increases stationarily to $T$ such that $\widetilde{Z}^{T_{n}}$ is a true martingale. Therefore, since $\widetilde Z{\cal E}(\theta\cdot S)$ is a supermartingale, by putting $\widetilde{Q}_{n}:=\widetilde{Z}_{T_{n}}\cdot P$ and using Jensen's inequality we derive
$$\begin{array}{llll}
\displaystyle E^{\widehat{Q}}\left[\left({\cal E}_{t\wedge T_{n}}(\theta\cdot S)/{\cal E}_{t\wedge T_{n}}(\widehat{\theta}\cdot S)\right)^{\displaystyle{{p}\over{2}}}\big|\mathcal{F}_{s}\right]\geq \left(E^{\widehat{Q}}\left[{\cal E}_{t\wedge T_{n}}(\theta\cdot S)/{\cal E}_{t\wedge T_{n}}(\widehat{\theta}\cdot S)\big |\mathcal{F}_{s}\right]\right)^{\displaystyle{{p}\over{2}}}\\
\hskip 6.6cm =\displaystyle \left(\frac{E^{\widetilde{Q}_{n}}\left({\cal E}_{t\wedge T_{n}}(\theta\cdot S)|\mathcal{F}_{s}\right)}{{\cal E}_{s\wedge T_{n}}(\widehat{\theta}\cdot S)}\right)^{\displaystyle{{p}\over{2}}}\geq \left( \frac{{\cal E}_{s\wedge T_{n}}(\theta\cdot S)}{{\cal E}_{s\wedge T_{n}}(\widehat{\theta}\cdot S)}\right)^{\displaystyle{{p}\over{2}}},\end{array}$$
for $0\leq s<t\leq T$. This proves that $\left({\cal E}(\theta\cdot S)/{\cal E}(\widehat{\theta}\cdot S)\right)^{p/2}$ is a nonnegative $\widehat{Q}$-local submartingale. Then, due to (\ref{p<0-integrability}) and de la Vall\'ee Poussin's argument, we deduce that this process is a true $\widehat{Q}$-submartingale. Again, an application of Jensen's inequality leads to,
$$
E^{\widehat{Q}}\left[\left({\cal E}_{t}(\theta\cdot S)/{\cal E}_{t}(\widehat{\theta}\cdot S)\right)^{p}\big|\mathcal{F}_{s}\right]\geq \left(E^{\widehat{Q}}\left[\left({\cal E}_{t}(\theta\cdot S)/{\cal E}_{t}(\widehat{\theta}\cdot S)\right)^{p/2}\big|\mathcal{F}_{s}\right]\right)^{2}\geq \left({\cal E}_{s}(\theta\cdot S)/{\cal E}_{s}(\widehat{\theta}\cdot S)\right)^{p}.
$$
Hence, $-\displaystyle\left({\cal E}(\theta\cdot S)/{\cal E}(\widehat{\theta}\cdot S)\right)^{p}$ is a $\widehat Q$-supermartingale and this ends the proof. \qed\\

In the following, we discuss some useful properties of the minimal Hellinger martingale densities (MHM-densities hereafter) under local change of probabilities.

\section{MHM-densities under Local Change of Measure}


In this last part of the Appendix, we will extent some properties of the minimal Hellinger martingale density to the case where one is facing a local change of probability. This extension is claimed by our parametrization approach of Sections \ref{powercase} and \ref{subsection:logcase}. Through out the rest of the paper, we consider a real number, $p\in{{\it \hbox{I\kern-.2em\hbox{R}}}}$, $p\not=1$, $q={{p}\over{p-1}}$, and a positive local martingale, $Z$, given
by \begin{equation}\label{defLocalDensity}Z:={\cal
E}(N),\,\,N:=\beta\cdot
S^{c}+W\star(\mu-\nu)+g\star\mu+\overline{N},\,\,W_{t}(x):=f_{t}(x)-1+\frac{\widehat{f_{t}}-a_t}{1-a_{t}}I_{\{a_{t}<1\}}.\end{equation}
Here $\left(\beta,f,g,\overline{N}\right)$ are the Jacod's
components of $N$. Throughout this section, we will frequently use the set of martingale density with respect to the density $Z$ defined by
\begin{equation}\label{setLocalMartingaleDensity}
\mathcal{Z}^{e}_{q,loc}(S,Z):=\Bigl\{\overline{Z}\ \ \ \mbox{process}\ \ \Big|\ \ \ \overline{Z}>0,\ \ \ \ \ \ \ \ \ \ \ \overline{Z}Z\in{\cal
Z}^{e}_{q,loc}(S)\Bigr\}, \end{equation} where ${\cal Z}^{e}_{q,loc}(S)$ is given
by (\ref{setofdernsities})-(\ref{sigmadensities}). Then, the minimal Hellinger martingale density of order $q$ with respect to $Z$ is given by the following.

Under $Z$ (the local change of measure), the process $S$ posses the following predictable characteristics and continuous local martingale part $S^{c,Z}$, $b^{Z}$, $a^{Z}$, $\nu^{Z}$ and $F^{Z}$ given by
\begin{equation}\label{charateristicDenisty}\begin{array}{llll}
S^{c,Z}:=S^{c}-c\beta\cdot A,\ b^{Z}:=b+c\beta+\displaystyle\int f(x)h(x)F(d x),\ a^{Z}_{t}:=\nu^{Z}(\{t\},\hbox{I\kern-.18em\hbox{R}}^{d}),\\
\\
\nu^{Z}(d t,d x):=F^{Z}_t(d x)d A_{t},\ \ F^{Z}_{t}(d x):=(1+f_{t}(x))F_{t}(d x).\end{array}\end{equation}

\begin{proposition}\label{lem-represofMHE-changemeasure}
Let $p\in (-\infty,1)$ and $q$ its conjugate number (i.e $ q:=p/(p-1)$). Consider a positive local martingale, $Z$, and $\widetilde Z\in{\cal
Z}^e_{q,loc}(S,Z)$. Suppose that there exists $\widetilde\beta\in L(S)$ such that $P\otimes A-a.e.$ $\widetilde\beta\in {\cal D}^+$ , root of
\begin{equation}\label{martingaleunderZ}
0\displaystyle=b^{Z}+(p-1)c\lambda+\int\Bigl[(1+\lambda^{tr}x)^{p-1}-1\Bigr]xF^{Z}(d x),\end{equation}
and
\begin{equation}\label{IntegrabilityofWunderZ}
\begin{array}{llll}
W\in {\cal G}^1_{loc}(\mu, Z),\ \ \ \ \ \displaystyle W_t(x):=\displaystyle{{(1+\widetilde\beta^{tr}_t x)^{p-1}-1}\over{1-a+\int u_t(x)\nu^Z(\{t\},dx)}},\\
\\
\mbox{and}\ \ \ \widetilde Z:={\cal E}(\widetilde N),\ \ \ \ \widetilde N:=(p-1)\widetilde\beta\cdot S^{c,Z}+W\star\left(\mu-\nu^Z\right).\end{array}\end{equation}
Then, the following hold.
\begin{equation}\label{ZtildeZ}\begin{array}{lll}
\widetilde Z^{q-1}={\cal E}\left(\widetilde H^{Z}\cdot
S+q(q-1)h^{(q)}(\widetilde Z, Z)\right)={\cal E}\left(\widetilde\beta\cdot
S\right){\cal E}\left(q(q-1)h^{(q)}(\widetilde Z, Z)\right).
\end{array}\end{equation}
Here, $\widetilde H^{Z}:=(\gamma^{Z})^{1-q}\widetilde\beta$ and $\gamma^{Z}:=1-a^{Z}_{t}+\int(1+\widetilde{\beta}^{tr}y)^{p-1}\nu^{Z}(\{t\},dy)$ satisfies
\begin{equation}\label{HZtilde}
\gamma^{Z}_{t}=\left(1+q(q-1)\Delta h^{(q)}(\widetilde Z, Z)\right)^{1-p},\end{equation}
and on $\{\Delta A=0\}$ we have
\begin{equation}\label{HellingerAc}
q(q-1)d h^{(q)}(\widetilde Z, Z)={{p}\over{2}}\widetilde{\beta}^{tr} c\widetilde{\beta}dA+\int\left[ (1+\widetilde\beta^{tr} x )^p-1-q(1+\widetilde{\beta}_t^{tr}x)^{p-1}-1)\right]F^Z(dx)d A.\end{equation}
\end{proposition}

{\bf Proof.} Since both the assumptions and the results of the proposition are stable under localization. In other words, if there exists a sequence of stopping time $(T_n)_n$ that increases to infinity such that the proposition is valid on $\Rbrack 0,T_n\Lbrack$ for each $n\geq 1$, then it will be valid globally. Therefore, there is no loss of generality in assuming that $Z$ is a martingale, and put $Q:=Z_T/Z_0\cdot P\sim P$. Then, it is easy to see that $(b^Z,c,F^Z)$ coincide with the predictable characteristics of $S$ under $Q$, $\nu^Z$ coincides with the $Q$-compensator of $\mu$, and the equations (\ref{martingaleunderZ})--(\ref{IntegrabilityofWunderZ}) translate into the fact that $\widetilde Z$ is the minimal Hellinger martingale density under $Q$ exists. Thus, a direct application of Corollary 4.7 of Choulli et al. (2007) (see also Theorem 3.1 of Choulli and Stricker(2009) for complete and other characterizations), leads to the first equality in (\ref{ZtildeZ}). The second equality in (\ref{ZtildeZ}) follows from Yor's formula (i.e. ${\cal E}(X){\cal E}(Y)={\cal E}(X+Y+[X,Y])$) and (\ref{HZtilde}). Hence, in the remaining part of this proof, we will focus on proving (\ref{HZtilde}) and (\ref{HellingerAc}). To this end, we start by calculating the Hellinger process for $\widetilde Z$ as follows.
\begin{equation}\label{HellingerforZtilde}
\begin{array}{llll}
 q(q-1)h^{(q)}(\widetilde Z, Z)= q(q-1) h^{(q)}(\widetilde Z, Q)= q(q-1){{(p-1)^2}\over{2}}\widetilde\beta^{tr}c{\widetilde\beta}\cdot A+\\
 \\
\hskip 2cm +
 \displaystyle \int \left[(1+\widetilde\beta^{tr} x)^{q(p-1)}(\widetilde\gamma_Q)^{-q}-1+q-q(1+\widetilde\beta^{tr} x)^{p-1}/\widetilde\gamma_Q\right]F(dx)\cdot A+\\
 \\
\hskip 2cm +
 \displaystyle \sum (1-a^Q)\left[(\widetilde\gamma_Q)^{-q}-1-q((\widetilde\gamma_Q)^{-1}-1)\right].\end{array}\end{equation}
 Therefore, since on the set $\{\Delta A=0\}$ we have $\widetilde\gamma_Q=1$, we deduce that (\ref{HellingerAc}) follows immediately from the above equation. By taking the jumps in both sides of (\ref{HellingerforZtilde}), and using the fact that $\Delta A c=0$, $b_t\Delta A_t=\int x F_t^Q(dx)\Delta A_t=\int x\nu^Q(\{t\},dx)$ and $
 \int x(1+\widetilde\beta^{tr} x)^{q(p-1)}\nu^Q(\{t\},dx)=0$ (that is derived from (\ref{martingaleunderZ})), we get 
 $$\begin{array}{llll}
  q(q-1)\Delta h^{(q)}(\widetilde Z, Z)=\widehat u^{Q}(1-a^Q+\widehat u^{Q})^{-q}-a^Q-q\left(\widehat u^{Q}/(1-a^Q+\widehat u^{Q})-a^Q\right)+\\
  \\
  \hskip 4cm +(1-a^Q)\left[(1-a^Q+\widehat u^{Q})^{-q}-1-q((1-a^Q+\widehat u^{Q})^{-1}-1)\right],
 \end{array} $$
 where $\widehat u^{Q}=\int (1+\widetilde\beta^{tr} x)^{p-1}\nu^Q(\{t\},dx)=\widetilde\gamma_Q-1+a^Q$. Thus, after simplifying the above obtained equation, we get (\ref{HZtilde}), and the proof of the proposition  is completed.\qed





\vskip 1cm 
\centerline{\bf\Large References}
\vskip 0.35cm
Anthropelos, M. (2013): Forward Exponential Performances: Pricing and Optimal Risk Sharing. Preprint. http://arxiv.org/abs/1109.3908.\\

Berrier, F., and L. C. G. Rogers. Tehranchi, M.R. (2009): A Characterization of Forward Utility Functions. mimeo, Statistical Laboratory, University of Cambridge.\\

Clark, S. A. (1996): The random utility model with an infinite choice space, Economic Theory, \textbf{7}(1), 179--189.\\

Cohen, M. A., Random utility systems -- the infinite case, Journal of Mathematical Psychology, \textbf{2}2(1), 1--23, 1980.\\

Choulli, T. Deng, J. and Ma, J. (2013): How Non-Arbitrage, Viability and Num\'eraire Portfolio are Related. Preprint of University of Alberta (25 pages).\\

Choulli, T. and Ma, J. (2013): An Optimal Deflator and Forward Utilities. Work in progress (25 pages).\\

Choulli, T., Ma, J. and Morlais, M. (2011): Three essays on exponential hedging with variable exit times. Musiela Festisch, Springer.\\

Choulli, T. and Stricker, Ch. (2009): Comparing the minimal Hellinger martingale measure of order $q$ and the $q$-optimal martingale measure. Stochastic Processes and their Applications 19, 1368--1385.\\

Choulli, T. and Stricker, C. (2005):
 Minimal entropy-{H}ellinger martingale measure in incomplete markets. Math. Finance \textbf{15}(3), 465--490.\\

Choulli, T. and Stricker, C. (2006):
 More on minimal entropy-{H}ellinger martingale measures.
 Math. Finance \textbf{16}(1), 1--19.\\

Choulli, T., Stricker, C., and Li, J. (2007):
 Minimal {H}ellinger martingale measures of order {$q$}.
 Finance Stoch. \textbf{11}(3), 399--427.\\

Cvitanic, J. and Karatzas, I. (1992): Convex duality in constrained portfolio optimization, The Annals of Applied Probability : 767-818.\\

Cvitanic, J., Schachermayer, W., and Wang, H. (2001): Utility maximization in incomplete markets with random endowment." Finance and Stochastics 5,2, 259--272.\\

Dellacherie, C. and Meyer, P-A. (1980): Th\'eorie des martingales. Chapter V to
VIII. Hermann.\\

Evans, J., Henderson, V., and Hobson, D. ( 2008): Optimal timing for an
indivisible asset sale. Math. Finance \textbf{18}, 545--567.\\

Fisher, I. (1931): The impatience theory of interest. AER.\\

G\"oll, T., and Kallsen, J. (2003): A complete explicit solution to the log-optimal portfolio problem. The Annals of Applied Probability 13,2, 774--799.\\

Hakansson, N. H. (1969): Optimal investment and consumption strategies
under risk, an uncertain lifetime and insurance. Int. Econ. Rev.
\textbf{10}, 443--466.\\

Harrison, J. M., and Pliska, S.R.(1981) "Martingales and stochastic integrals in the theory of continuous trading." Stochastic processes and their applications 11,3, 215-260.\\

He, S., Wang, J.G. and Yan, J.A. ((1992): Semimartingale theory and stochastic calculus, Science Press and CRC Press, Beijing.\\

Henderson, V.: Valuing the option to invest in an incomplete market.
Math. Finance Econ. \textbf{1}(2), 103--128, 2007.\\

Henderson, V. and Hobson, D. (2007): Horizon-unbiased utility functions.
Stoch. Process. Appl. \textbf{117}(11), 1621--1641.\\

Jacod, J. (1979): Calcul stochastique et probl\`emes de martingales. no.
714 in Lecture
  Notes in Mathematics. Springer, Berlin .\\

Jacod, J. and Shiryaev, A.: Limit Theorems for Stochastic Processes,
2ed edn. Springer, 2002.\\

Kabanov, Y. M. (2013): On local martingale deflators and market portfolios: The Karadaras Theorem, private communication.\\

Kabanov, Y. M., Liptser, R. S., Shiryaev, A. N. (1984): On the proximity in variation of probability measures, Dokl. Akad.
Nauk SSSR, 278, 2, 1984. English translation: Soviet Math. Dokl., \textbf{30}(2).\\

Kabanov, Y. M. (1985): An estimate of closeness in variation of probability measures,
Probab, Theory and Its Appl., \textbf{30}(2).\\

Kabanov, Y. M., Liptser, R. S., Shiryaev, A. N.(1986): On the variation distance for probability measures defined on a filtered space, Probability theory and related fields, \textbf{71}(1), 19--35.\\

Karatzas, I. and Wang, H. (2000): Utility maximization with discretionary
stopping. SIAM J. Control Optim. \textbf{39}(1), 306--329.\\

Karatzas, I., and Zitkovic, G. (2003): Optimal consumption from investment and random endowment in incomplete semimartingale markets." The Annals of Probability 31,4, 1821--1858.\\

Kardaras, C.(2012): {Market viability via absence of arbitrage of the first kind.}  Finance and Stochastics, 1-17 (2012).\\

Kramkov, D., Schachermayer, W. (1999): 
\newblock The asymptotic elasticity of utility functions and optimal
investment in incomplete markets, The Annals of Applied Probability, {\bf 9}(3), 904--950.\\

Ma, J. (2013): Minimal Hellinger Deflators and HARA Forward Utilities with Applications: Hedging with Variable Horizon, PhD thesis, University of Alberta.\\

McFadden, D. and Richter, M. K. (1990): Stochastic rationality and revealed stochastic preference, Preferences, uncertainty, and optimality, essays in honor of Leo Hurwicz, 161--186.\\

Merton, R.C. (1971): 
\newblock Optimum consumption and portfolio rules in a continuous-time model, Journal of Economic Theory, Elsevier, {\bf 3}(4), 373-413.\\

Merton, R.C. (1973): Theory of rational option pricing, The Bell Journal of Economics and Management Science, 141-183.\\

 Musiela, M. and Zariphopoulou, T. (2007): Investment and valuation under backward and forward dynamic exponential utilities in a stochastic factor
 model. In: Fu, M., Jarrow, R., Yen, J., Elliot, R. (eds.) Advances in
Mathematical Finance, pp. 303--334. Birkhauser, Boston.\\

 Musiela, M. and Zariphopoulou, T. (2009a): Backward and forward utilities and the associated indifference
pricing systems: The case study of the binomial model. In: Carmona,
R. (ed.) Indifference Pricing, pp. 3--43. Princeton University Press.\\

Musiela, M. and Zariphopoulou, T. (2009b): Portfolio choice under dynamic
investment performance criteria. Quant. Finance \textbf{9}(2),
161--170.\\

Musiela, M. and Zariphopoulou, T. (2010): Portfolio choice under space-Time
monotone performance criteria. SIAM J. Financial Math. \textbf{1},
326--365.\\

%
%


Rockafellar, R.T. (1970): Convex Analysis. Princeton University Press,
Princeton.\\

Suppes, P., Krantz, D. H., Luce, R. D.,  Tversky, A. ((1989) Foundations of Measurement, vol. II. New York: Academic Press.\\

Yaari, M.E. ((1965): Uncertain lifetime, life insurance and the theory of
the consumer. Rev. Econ. Stud. \textbf{32}, 137--158.\\

Zariphopoulou, T. and Zitkovic, G. (2010): Maturity-independent risk measures." SIAM Journal on Financial Mathematics, 1,1, 266--288.\\

Zitkovic, G. (2009): A dual characterization of self-generation and log-affine forward performances. Ann. Appl.
 Prob.
 \textbf{19}(6), 2176--2270.\\


\end{document}